\documentclass[lettersize,journal]{IEEEtran}
\IEEEoverridecommandlockouts

\usepackage[final]{changes}
\newcommand{\stkout}[1]{\ifmmode\text{\sout{\ensuremath{#1}}}\else\sout{#1}\fi}
\setdeletedmarkup{\stkout{#1}}

\usepackage{cite}
\usepackage{amsmath,amssymb,amsfonts}
\usepackage{algorithmic}
\usepackage{graphicx}
\usepackage{textcomp}
\usepackage{xcolor}
\def\BibTeX{{\rm B\kern-.05em{\sc i\kern-.025em b}\kern-.08em
    T\kern-.1667em\lower.7ex\hbox{E}\kern-.125emX}}

\usepackage{mathrsfs}
\usepackage{amsthm}
\usepackage{verbatim}
\usepackage{dsfont}
\usepackage{bm}

\usepackage{caption}

\newtheorem{theorem}{Theorem}
\newtheorem{lemma}{Lemma}
\newtheorem{corollary}{Corollary}

\theoremstyle{definition}
\newtheorem{remark}{Remark}
\newtheorem{definition}{Definition}
\newtheorem{example}{Example}

\usepackage{algorithm}
\usepackage{bbm}
\usepackage{hyperref}
\hypersetup{colorlinks = true
            linkcolor=red,
            anchorcolor=blue,
            citecolor=green}
              
\newcommand{\MF}{\mathcal{F}}
\newcommand{\MI}{\mathcal{I}}
\newcommand{\MR}{\mathcal{R}}
\newcommand{\mM}{\mathcal{M}}

\newcommand{\FF}{\mathbb{F}}

\begin{document}

\title{On the Distribution of Weights Less than $2w_{\min}$ in Polar Codes}

 \author{
   \IEEEauthorblockN{Zicheng Ye\IEEEauthorrefmark{2}\IEEEauthorrefmark{3}\IEEEauthorrefmark{1}, 
   Yuan Li\IEEEauthorrefmark{2}\IEEEauthorrefmark{3}\IEEEauthorrefmark{1}, 
   Huazi Zhang\IEEEauthorrefmark{1}, 
   Jun Wang\IEEEauthorrefmark{1}, 
   Guiying Yan\IEEEauthorrefmark{2}\IEEEauthorrefmark{3}, 
   and Zhiming Ma\IEEEauthorrefmark{2}\IEEEauthorrefmark{3} } \\
  \IEEEauthorblockA{\IEEEauthorrefmark{2}
                     School of Mathematical Sciences, University of Chinese Academy of Sciences} \\
   \IEEEauthorblockA{\IEEEauthorrefmark{3} 
                     Academy of Mathematics and Systems Science, CAS } \\
  \IEEEauthorblockA{\IEEEauthorrefmark{1}
                     Huawei Technologies Co. Ltd.} \\
  Email: \{yezicheng3, liyuan299, zhanghuazi, justin.wangjun\}@huawei.com,\\
           yangy@amss.ac.cn, mazm@amt.ac.cn  
           \thanks{This work was supported by the National Key R\&D Program of China (2023YFA1009602).}
           }

\maketitle

\begin{abstract}
The number of low-weight codewords is critical to the performance of error-correcting codes. In 1970, Kasami and Tokura characterized the codewords of Reed-Muller (RM) codes whose weights are less than $2w_{\min}$, where $w_{\min}$ represents the minimum weight. In this paper, we extend their results to decreasing polar codes. We present the closed-form expressions for the number of codewords in decreasing polar codes with weights less than $2w_{\min}$. Moreover, the proposed enumeration algorithm runs in polynomial time with respect to the code length.
\end{abstract}


\section{Introduction}
\IEEEPARstart{P}{olar} codes \cite{Arikan2009}, introduced by Ar{\i}kan, are a significant breakthrough in coding theory. As the code length approaches infinity, polar codes can approach channel capacity under successive cancellation (SC) decoding. For short to moderate block lengths, successive cancellation list (SCL) decoding \cite{Niu2012, Tal2015} can significantly improve the error-correcting performance. In fact, SCL algorithms can approach the maximum likelihood (ML) decoding performance when the list size is large enough.

The weight distribution of linear codes has a significant impact on the ML decoding performance, which can be estimated accurately through the union bound with the number of low-weight codewords \cite{Sason2006}. Due to the importance of the weight distribution, many researches have approximated the weight distribution of different codes \cite{Krasikov1995, Di2006, Kasami1985}. However, in the general case, the complexity of computing the exact weight distribution grows exponentially with the code length. 

For polar codes, since the number of low-weight codewords affects the performance of SCL decoding with a large list size, there have been many attempts to study the weight distribution of polar codes. In \cite{Li2012}, SCL decoding with a very large list size at a high SNR was proposed to collect the low-weight codewords. This method was improved in \cite{Liu2014} to save memory. In \cite{Valipour2013, Zhang2017}, polynomial-complexity probabilistic approaches were proposed to approximate the weight distribution of polar codes. The authors in \cite{Yao2023} designed an algorithm to calculate the exact weight distribution of original or specific pre-transformed polar codes \cite{Li2019} by coset method, but the complexity is still exponential. In \cite{Li2021, Li2023}, the authors proposed efficient recursive formulas to calculate the average weight spectrum of pre-transformed polar codes with polynomial complexity.

RM codes \cite{Reed1953, Muller1954} are closely related to polar codes. Existing researches on the weight distribution of RM codes may shed light on the characterization of polar code weight spectrum. In \cite{Kasami1970}, the numbers of codewords with weights less than $2w_{\min}$ in RM codes were determined, and the result was then generalized to $2.5w_{\min}$ in \cite{Kasami1976}, where $w_{\min}$ is the minimum weight.
  
The authors in \cite{Bardet2016} regarded polar codes as decreasing monomial codes and used the lower triangular affine transformation automorphism to calculate the number of codewords with weight $w_{\min}$. Recently, the same method was applied to calculate the number of codewords with weight $1.5w_{\min}$ \cite{Rowshan2023}.\deleted{However, it seems difficult to further extend this method to the number of codewords with other weights.}\added{Their methods rely on the observation that the codewords with weight $w_{\min}$ and $1.5w_{\min}$ can be obtained by lower triangular affine transformations of a single row or a summation of two rows. However, this property no longer holds for codewords with weight larger than 1.5$w_{\min}$.} Therefore, we propose a unified  method to calculate the number of codewords with weight less than $2w_{\min}$. \added{A detailed comparison between our method and the existing results is provided in subsection \ref{ss24}.}

In this paper, we generalize the results in \cite{Kasami1970} and provide closed-form expressions for the number of codewords with weights less than $2w_{\min}$ in decreasing polar codes. A decreasing polar code is a subcode of some RM code, possessing the same minimum weight. In brief, our task is to select those low-weight codewords that belong to the polar code. However, it is worth noting that the enumeration procedure for polar codes is more complex compared to that of RM codes. We divide codewords into disjoint subsets based on the largest terms of their monomial representations. The size of subsets grows logarithmically with the code length, while the number of subsets is less than the square of code length. As a result, the time complexity of the enumeration algorithm is almost proportional to the square of code length. 

The rest of this paper is organized as follows. In Section II, we provide a concise introduction to polar codes and RM codes. In addition, we outline our proof. In section III, we classify and enumerate the codewords with weight less than $2w_{\min}$, and provide algorithms to compute the weight distribution. In section IV, the closed-form expressions for the number of codewords and the algorithm complexity is presented. Finally, in section V, we draw some conclusions.

\section{Preliminaries}

\subsection{Polar codes as monomial codes}
Let $\replaced{G}{F}=\begin{bmatrix} 1&0 \\ 1&1 \end{bmatrix}$ and $\replaced{G_N}{F_N}=\replaced{G}{F}^{\otimes m}$, where $\otimes$ is Kronecker product:
$$
A \otimes B =\left[\begin{array}{ccc}
a_{11} B & \cdots & a_{1 k} B \\
\vdots & \ddots & \vdots \\
a_{t 1} B & \cdots & a_{t k} B
\end{array}\right].
$$
Polar codes can be constructed by selecting $K$ rows of $G_N$ as the information set $\mathcal{I}$. $\MF = \MI^c$ is called the frozen set. Denote the polar code with information set $\MI$ by $C(\MI)$. 

Polar codes can be described as monomial codes \cite{Bardet2016}. The monomial set is denoted by
$$
\mM \triangleq \{x_1^{a_1}...x_{m}^{a_{m}}\mid (a_1,...,a_{m})\in\FF_2^m\}.
$$
Let $e = x_1^{a_1}...x_{m}^{a_m}$ be a monomial in $\mM$, the degree of $e$ is defined as $\text{deg}(e) = \sum_{i=1}^m a_i$. \added{In particular, $1$ is a monomial with degree zero.} The polynomial set is denoted by
$$
\MR_{\mM} \triangleq \{\sum_{e\in \mM} a_e e \mid a_e\in\FF_2\},
$$
and the degree of the non-zero polynomial $g =\sum_{e\in \mM} a_e e$ is defined as $\text{deg}(g) = \max_{a_e\neq 0} \{\text{deg}(e)\}$, and the zero polynomial has no degree. 

A polynomial $g$ is said to be linear if deg$(g) = 1$. We say a linear polynomial $g$ is a linear factor of \added{a non-zero polynomial} $f$ if $f = gh$\added{, where $h$ is a polynomial}. Linear polynomials $g_1,...,g_i$ are said to be linearly independent if $a_1g_1+\cdots+a_ig_i+a_0=0$ implies $a_0=a_1=\cdots=a_i=0$. 

\begin{example}
$x_1,\cdots,x_m$ are linearly independent polynomials. $x_1,x_2,x_1+x_2+1$ are not linearly independent since $x_1+x_2+(x_1+x_2+1)+1=0$.
\end{example}

For a linear polynomial $f$, define the largest term of $f$ to be the largest $x_i$ appearing in $f$, and define $F(f)$ to be its index, i.e., if $f = \sum_{i=1}^m a_ix_i + a_0$, $F(f)=\max\{i\mid a_i\neq 0\}$.  Notice that linear polynomials with different largest terms are linearly independent.

Let $b\in\FF_2^m$. Denote $g(b)$ to be the evaluation of $g$ at point $b$.
We say a linear polynomial $g$ is a linear factor of $f$ if $f = gh$, i.e., for any $b\in\FF_2^m$, $g(b)=0$ implies $f(b)=0$.

The length-$N=2^m$ evaluation vector of $g\in \MR_{\mM}$ is denoted by
$$
\replaced{\text{ev}}{\text{eval}}(g) \triangleq (g(x))_{x\in\FF_2^m}.
$$
\added{And the weight of $g$ is defined as the Hamming weight of $\text{ev}(g)$.}

In fact, each row of $\replaced{G_N}{F_N}$ can be expressed by $\text{ev}(e)$ for some $e\in\mM$. To see this, for each $z\in \{0,1,...,2^m-1\}$, there is a unique binary $a=(a_1,...,a_m)$ of $2^m-1-z$, where $a_1$ is the least significant bit, such that
\[\sum_{i=1}^{m} 2^{i-1}a_i = 2^m-1-z.\]
Then we can check that $\text{ev}(x_1^{a_1}...x_m^{a_m})$ is exactly the $(2^m-z-1)-th$ row of $\replaced{G_N}{F_N}$. As seen, the three forms, i.e., the integer $z$, the binary representation of $2^m-1-z = (a_1,...,a_m)$ and the corresponding monomial $x_1^{a_1}...x_{m}^{a_{m}}$ all refer to the same thing. Therefore, the information set $\MI$ can be regarded as a subset of the monomial set $\mM$.

\begin{definition}
Let $\MI$ be a set of monomials. The monomial code $C(\MI)$ with code length $N = 2^m$ is defined as
$$
C(\MI) \triangleq \operatorname{span} \{\operatorname{ev}(e): e \in \MI\}.
$$
\end{definition}

If the maximum degree of monomials in $\MI$ is $r$, we say the monomial code $C(\MI)$ is $r$-$th$ order. As mentioned, we use the polynomials in $ \MR_{\mM}$ to represent the corresponding codewords in monomial codes.

\begin{example}

The example shows the row vector representations of $F_8$.

$$
\begin{array}{lc}
\begin{array}{c} x_1= \\ x_2=\\ x_3 =\\ \hline x_1x_2x_3 \\ x_2x_3 \\ x_1x_3 \\ x_3 \\ x_1x_2\\ x_2 \\ x_1 \\ 1
\end{array}
\begin{array}{cccccccc}
1 &0 &1 &0 &1 &0 &1 &0\\
1 &1 &0 &0 &1 &1 &0 &0\\
1 &1 &1 &1 &0 &0 &0 &0 \\
\hline
1 & 0 & 0 & 0 & 0 & 0 & 0 & 0\\
1 & 1 & 0 & 0 & 0 & 0 & 0 & 0\\
1 & 0 & 1 & 0 & 0 & 0 & 0 & 0\\
1 & 1 & 1 & 1 & 0 & 0 & 0 & 0\\
1 & 0 & 0 & 0 & 1 & 0 & 0 & 0\\
1 & 1 & 0 & 0 & 1 & 1 & 0 & 0\\
1 & 0 & 1 & 0 & 1 & 0 & 1 & 0\\
1 & 1 & 1 & 1 & 1 & 1 & 1 & 1
\end{array}
\end{array}
$$

If the information set is $\MI=\{1,x_1,x_2,x_3\}$, then the generator matrix of the polar code $C(\MI)$ is

$$
\begin{bmatrix}
1 & 1 & 1 & 1 & 0 & 0 & 0 & 0\\
1 & 1 & 0 & 0 & 1 & 1 & 0 & 0\\
1 & 0 & 1 & 0 & 1 & 0 & 1 & 0\\
1 & 1 & 1 & 1 & 1 & 1 & 1 & 1
\end{bmatrix}
$$

\end{example}

\subsection{Decreasing monomial codes}

The partial order of monomials was defined in \cite{Bardet2016} and  \cite{Schurch2016}. Two monomials with the same degree are ordered as $x_{i_1}...x_{i_t}\preccurlyeq x_{j_1}...x_{j_t}$ if and only if $i_l \leq j_l$ for all $l\in\{1,...,t\}$, where we assume $i_1 <...< i_t$ and $j_1 <...< j_t$. This partial order is extended to monomials with different degrees through divisibility, namely $e \preccurlyeq e'$ if and only if there is a divisor $e''$ of $e'$ such that $e \preccurlyeq e''$. In fact, $e \preccurlyeq e'$ means $e$ is universally more reliable than $e'$. 

An information set $\MI\subseteq \mathcal{M}$ is decreasing if $\forall e\preccurlyeq e'$ and $e'\in\MI$ we have $e \in \MI$. A decreasing monomial code $C(\MI)$ is a monomial code with a decreasing information set $\MI$. If the information set is selected according to the Bhatacharryya parameter or the polarization weight (PW) method \cite{He2017}, the polar codes are decreasing. Decreasing polar codes $C(\MI)$ can be generated by $\MI_{\text{min}}$, when $\MI$ is the smallest decreasing set containing $\MI_{\text{min}}$, i.e., $\MI = \{e\in \mM \mid \exists\ e'\in \MI_{\text{min}}, e\preccurlyeq e'\}$. From now on, we always suppose $\MI$ is decreasing.

\subsection{Weight distribution of RM codes} 

\begin{definition}
The $r$-$th$ order RM code $\operatorname{RM}(m, r)$ is defined as
$$
\operatorname{RM}(m, r) \triangleq \operatorname{span}\left\{\operatorname{ev}(e): e \in \mM, \operatorname{deg}(e) \leq r\right\}.
$$
\end{definition}

Clearly, RM codes are decreasing monomials codes. It is well-known that the minimum weight of non-zero polynomials in $\operatorname{RM}(m, r)$ is $2^{m-r}$.

The number of codewords with weight $i$ in $C$ is denoted as
$$
A_{i}(C)  \triangleq |\{c \in C \mid \mathrm{wt}(c)=i\}|.
$$
The sequence $\left(A_{i}(C)\right)_{i=0}^{N}$ is called the weight distribution of $C$. 

The weight distribution as well as the polynomials in RM codes with weight less than $2w_{\min}$ is presented in \cite{Kasami1970}.

\begin{theorem}[\cite{Kasami1970}] \label{thm1}

Let $f$ be a non-zero polynomial in $\operatorname{RM}(m, r)$. If the weight of $f$ is less than $2^{m-r+1}$, then $\mathrm{wt}(f) = 2^{m-r+1}-2^{m-r+1-\mu}$, where $\mu$ is a positive integer. Moreover, $f$ can be written as one of the following forms:
\begin{align}\label{eq1}
& P(g_1,\cdots, g_{r+2 \mu-2})  \triangleq  \notag \\
& g_1 \cdots g_{r-2}(g_{r-1} g_r + \added{g_{r+1} g_{r+2} +} \cdots+g_{r+2 \mu-3} g_{r+2 \mu-2}), \notag \\ 
& m-r+2  \geq 2 \mu \geq 2,
\end{align}
or
\begin{align}\label{eq2}
& Q(g_1,\cdots, g_{r+\mu})  \triangleq  \notag \\
& g_1 \cdots g_{r-\mu}(g_{r-\mu+1}  \cdots g_r + g_{r+1} \cdots g_{r+\mu}), \notag \\ 
& m \geq r+\mu, r \geq \mu \geq 3.
\end{align}
\replaced{Here}{here} $g_1,g_2,\cdots$ are linearly independent linear polynomials. We call the Equation (\ref{eq1}) the type-I polynomials, and Equation (\ref{eq2}) the type-II polynomials. 

\end{theorem}

\begin{example}
Let $m = 7$ and $r = 4$, the minimum weight $w_{\min}$ of $\operatorname{RM}(7, 4)$ is 8,

$f_1 = x_1x_2x_3x_4$, by Equation (\ref{eq1}), $\mu = 1, \mathrm{wt}(f_1) = 8$ ;

$f_2 = x_1x_2(x_3x_4+x_5x_6)$, by Equation (\ref{eq1}), $\mu = 2, \mathrm{wt}(f_1) = 12$;

$f_3 = x_1(x_2x_3x_4+x_5x_6x_7)$, by Equation (\ref{eq2}), $\mu = 3, \mathrm{wt}(f_1) = 14$ .
\end{example}

We say $P(g_1,\cdots, g_{r+2 \mu-2})$ and $P(g'_1,\cdots, g'_{r+2 \mu-2})$ (or $Q(g_1,\cdots, g_{r+\mu})$ and $Q(g'_1,\cdots, g'_{r+\mu})$) have different forms if \deleted{if} there exists some $i$ such that $g_i \neq g'_i$. It is important to note that different forms can lead to the same polynomial. For instance, consider the case when $m=4$ and $r=2$. Let $g_1 = x_1, g_2 = x_4, g_3 = x_2, g_4 = x_2,  g'_1 = x_1, g'_2 = x_4+x_2, g'_3 = x_2, g'_4 = x_3+x_2$. In this example, we have $P(g_1,g_2,g_3,g_4) = P(g'_1,g'_2,g'_3,g'_4)$, despite $g_2\neq g'_2$ and $g_4\neq g'_4$. This is one of the main challenges in the enumeration process.

\subsection{Weight distribution of polar codes}\label{ss24}

In this paper, we utilize Theorem \ref{thm1} to enumerate the polynomials with weight less than $2w_{\min}$ in decreasing polar codes. Note that an $r$-$th$ order decreasing monomial code $C(\MI)$ is a subcode of the $r$-$th$ RM code. 

\begin{example}\label{ex:0}
Let $C(\MI)$ be a [128,80] decreasing polar code where $\MI$ is generated by $\MI_{\min}=\{x_1x_2x_5x_7,x_1x_3x_4x_7,x_1x_4x_5x_6,x_2x_3x_5x_6,x_4x_6x_7\}$. By Theorem $\ref{thm1}$, the minimum weight of $C(\MI)$ is 8.

As an example, let us consider polynomials in $C(\MI)$ with weight less than 16. Since $C(\MI)$ is a subcode of RM$(7,4)$, according to Theorem $\ref{thm1}$, the polynomials \added{in $C(\MI)$ with weight less than 16 can be written as} \deleted{are} the type-I polynomials $P(g_1,\cdots, g_{r+2 \mu-2})$ where $r=4$ and $\mu =1,2$ or the type-II polynomials $Q(g_1,\cdots, g_{r+\mu})$ where $r=4$ and $\mu =3$. We select $g_1,g_2,\cdots$ in order to enumerate these polynomials. 

During the enumeration process, $g_1,g_2,\cdots$ must be linearly independent, as stated in Theorem $\ref{thm1}$. Besides, we must guarantee that the polynomial belongs to the subcode $C(\MI)$. 

Moreover, it is important to note that different forms may represent the same polynomial, and thus it is redundant to include all these forms in the enumerations. For example, let us consider $f = x_1x_2(x_3x_6+x_4x_5) =  P(x_1,x_2,x_3,x_6,x_4,x_5)$ and $f' = x_1x_2(x_3(x_6+x_4)+x_4(x_5+x_3)) = P(x_1,x_2,x_3,x_6+x_4,x_4,x_5+x_3)$. It is straightforward to confirm that $f=f'$, i.e., both $f$ and $f’$ represent the same polynomial. 
\end{example}

In conclusion, in order to enumerate the type-I and type-II polynomials in $C(\MI)$ accurately, we need to carefully consider the following points in the enumeration procedure:

1) The linear polynomials $g_1,g_2,\cdots$ should be linearly independent;

2) The polynomial $P(g_1,\cdots, g_{r+2 \mu-2})$ or $Q(g_1,\cdots, g_{r+\mu})$  must belong to $C(\MI)$;

3) Different forms may represent the same polynomial. Thus, it is important to eliminate the redundant forms during the enumeration.

The authors in \cite{Kasami1970} applied Theorem \ref{thm1} to calculate the number of polynomials with weight less than $2w_{\min}$ in RM codes. It should be noted that in their study, since all monomials with degree $r$ belong to RM$(m,r)$, the second point is automatically satisfied. The polynomials are selected carefully so that they do not become linear combinations of the previously selected polynomials. The third point only appears in the type-I polynomial, and the number is calculated based on the weight distribution of the second order RM codes \cite{Sloane1970}.

\added{Since $r$-$th$ order decreasing polar code $C(\MI)$ is a subcode of $r$-$th$ order RM code, the polynomials with weight less than $2w_{\min}$ can be expressed as the type-I polynomials or the type-II polynomials, so there is no concern about miscounting.} However, only a fraction of the type-I and type-II polynomials actually belong to $r$-$th$ order polar codes. This distinction gives rise to completely different methods to examine the three above points \added{to avoid overcounting} in the enumeration procedure for polar codes. \added{The first point is necessary, otherwise the polynomial is not a codeword with weight less than $2w_{\min}$. The second point ensures that we exclude polynomials not in $C(\MI)$, while the third point prevents counting the same polynomials multiple times.} To facilitate this examination process, the largest term of a linear polynomial can be utilized as a helpful tool, which will be explained in detail later.

\added{The case $\mu=1$, i.e., the number of polynomials with $w_{\min}$, was calculated in \cite{Bardet2016}, and the case $\mu=2$, with weight $1.5w_{\min}$, was addressed in \cite{Rowshan2023}. In these studies, they applied the lower-triangular affine automorphism group to partition the codewords into different orbits.}

\added{When $\mu=2$, the authors in \cite{Rowshan2023} provided a technique to avoid overcounting for $f + g = f' + g'$, where $f,g,f',g'$ are degree-2 polynomials. However, when $\mu\geq 3$, avoiding overcounting with the sum of $\mu$ polynomials is more complex, which cannot be solved by \cite{Rowshan2023}. To handle this problem, we present a novel method by dividing the low-weight polynomials into different sets based on their first terms rather than orbits. And we count the free coefficients in polynomials in Algorithm \ref{alg:1}. }

\added{Additionally, when $\mu\geq 3$, it is necessary to consider type-II polynomials, which is not involved in \cite{Rowshan2023}. To solve this problem, we also divide the low-weight polynomials into different sets based on their first terms. And we count the linearly independent polynomials in Algorithm \ref{alg:2}.}

\subsection{\added{Proof Outline}}

We describe our proof briefly as follows. In subsection \ref{ss31}, we demonstrate that every polynomial can be expressed in a restricted form that satisfies the conditions in Lemma \ref{lemma2}. Besides, when considering restricted forms, the linear factors of a polynomial are unique. Although repetitive counting has not been completely eliminated, this transformation reduces the number of possible forms, which simplifies subsequent enumeration significantly.

In subsection \ref{ss32}, we calculate the number of the type-I polynomials $P(g_1,\cdots, g_{r+2 \mu-2})$. Focusing on the restricted forms, we ensure that the largest terms of $g_1,\cdots, g_{r+2\mu-2}$ can be rewritten to be all distinct and unique, as proven in Lemma \ref{lemma3} and \ref{lemma4}. Therefore, we can classify the polynomials into subsets based on their largest terms. This prevents duplications and enables clear categorization.

To calculate the size of each subset, we choose the linear polynomials $g_1,\cdots, g_{r+2\mu-2}$ in order so that the three points outlined in Example \ref{ex:0} are satisfied. Checking the first two conditions is straightforward. The challenge lies in the fact that different forms may result in the same polynomial, as shown in Example \ref{ex:0}.

The number of choices for $g_1,\dots, g_{r-2}$ is \replaced{$2^{\lambda_{x_{F(g_1)}\cdots x_{F(g_{r-2})}}} = 2^{\sum_{t=1}^{r-2} (F(g_t)-t+1)}$}{$\lambda_{x_{F(g_1)}\cdots x_{F(g_{r-2})}} =  \sum_{t=1}^{r-2} (F(g_t)-t+1)$}. \deleted{Besides, we note that $\sum_{j=1}^{\mu} g_{r-3+2j}g_{r-2+2j}$ is a polynomial with degree 2, and some coefficients are restricted by the others. We calculate the number of free coefficients in Theorem \ref{thm:m1}.} \added{Next, the number of distinct polynomials $\sum_{j=1}^{\mu} g_{r-3+2j}g_{r-2+2j}$ is computed in Theorem \ref{thm:m1}, and the proof is divided into two parts. In the first part (Lemma \ref{lemma7}), we establish the feasibility of rewriting the polynomials while fixing certain coefficients of $g_{r-1},\cdots, g_{r+2\mu-2}$ to 0. Then the number of possible polynomials is upper bounded by $2^{\alpha_u + \beta_u + 2\gamma_u}$, where $\alpha_u + \beta_u + 2\gamma_u$ is the number of the remaining coefficients defined in subsection \ref{ss32}. In the second part, we illustrate that the remaining coefficients are actually free, i.e., any value of these coefficients yields a different polynomial. Therefore, the number of low-weight codewords is exactly $2^{\alpha_u + \beta_u + 2\gamma_u}$.}

In subsection \ref{ss33}, we calculate the number of the type-II polynomials $Q(g_1,\cdots, g_{r+\mu})$. Similar to the previous case, we focus on the restricted forms. As shown in Lemma \ref{lemma5}, \added{we can also classify the polynomials into subsets based on their largest terms. Similarly, we need to choose the linear polynomials to satisfy three points in Example \ref{ex:0}. By Lemma \ref{lemma5},} regardless of order, different restricted forms represent different polynomials. \added{Therefore, the third point is satisfied, and the challenge lies in the first and the second points. }\deleted{We can also classify the polynomials into subsets based on their largest terms.}

\deleted{To calculate the size of each subset, we choose the linear polynomials $g_1,\cdots, g_{r+\mu}$ that satisfy the first and second points in Example \ref{ex:0}.}  The number of choices for $g_1,\cdots, g_r$ is $\replaced{2^{\lambda_{x_{F(g_1)}\cdots x_{F(g_{r})}}}}{\lambda_{x_{F(g_1)}\cdots x_{F(g_{r})}}}$. We then choose $g_{r+1},\cdots,g_{r+\mu}$ in order so that the linearly independence condition is satisfied, i.e., $g_{r+j}$ is not the linear combination of $g_1,\cdots,g_{r+j-1}$ for all $1\leq j\leq \mu$. For the second point, the difficulty arises when $F(g_{r-\mu+j}) = F(g_{r+j})$ for all $1\leq j\leq \mu$. Since the largest term is cancelled, $x_{F(g_1)}\cdots x_{F(g_r)}\in \MI$ is not a necessary condition for $Q(g_1, \cdots, g_{r+\mu})\in C(\MI)$. In this case, the coefficients of monomials not belonging to $C(\MI)$ in $g_1\cdots g_r$ and $g_1\cdots g_{r-\mu}g_{r+1}\cdots g_{r+\mu}$ must be equal thus are cancelled and disappear in $Q(g_1, \cdots, g_{r+\mu})$. We calculate the number of choices for $g_{r+1},\cdots,g_{r+\mu}$ satisfying this condition in Theorem \ref{thm:m2}.

\section{Classification and enumeration of polynomials}

In this section, we focus on the classification and enumeration of the polynomials with weight less than $2w_{\min}$ in $r$-order decreasing monomial codes $C(\MI)$. According to Theorem \ref{thm1}, we know $w_{\min} = 2^{m-r}$ and there are two kind of polynomials with weight less than $2^{m-r+1}$: the type-I and the type-II polynomials. We enumerate the number of polynomials in each type respectively by dividing polynomials into disjoint sets and calculating the size of each set. 

\added{For $u = (i_1,\dots,i_l)$, define the number of indices $i_j$ with $j\in [s,t]$ and $i_j<k$ to be $\varphi_{u,[s,t]}(k) = |\{s\leq j\leq t  \mid i_j < k \}|$, and define $ \bar{\varphi}_{u,[s,t]}(k) = k - \varphi_{u,[s,t]}(k)$. This definition simplifies the enumeration of low-weight codewords.}

\subsection{Restricted form of polynomials}\label{ss31}

As illustrated in Example \ref{ex:0}, a polynomial can be expressed in various forms. However, we prove that every polynomial can be expressed in a way that satisfies the conditions stated in Lemma \ref{lemma2}. This allows us to reduce the number of possible forms, which is beneficial for the subsequent enumeration. 

\begin{lemma}\label{lemma1}
Let $f = gh$ where $g$ is a linear polynomial with $F(g) = t$. Construct $h'$ by replacing every $x_t$ in $h$ with $ g + x_t + 1$. Then $f = gh'$ and $x_t$ does not \replaced{appear}{appears} in $h'$.
\end{lemma}

\begin{proof}
Let $b=\added{(b_1,\dots,b_m)}\in\FF_2^m$. If $g(b) = 0$, then $0 = g(b)h(b) = g(b)h'(b)$. If $g(b)=1$, then $b_t = g(b) + b_t + 1$, so $h'(b) = h(b)$. Therefore, for any $b\in\FF_2^m$, $f(b) = g(b)h'(b)$.
\end{proof}

\begin{example}
Let $g = x_2 + x_1$ and $h = x_2 + x_2x_3$, then $h' = x_1+1 +(x_1+1)x_3$ by replacing every $x_2$ in $h$ by $ x_1 + 1$. Now $(x_2 + x_1)(x_2 + x_2x_3) = (x_2 + x_1)(x_1+1 +(x_1+1)x_3)$ and $x_2$ does not \replaced{appear}{appears} in $h'$.
\end{example}

\begin{lemma}\label{lemma1+}
Let $f = g$. $f',g'$ are the polynomials by replacing every $x_t$ in $f,g$ with another polynomial $h$. Then $f' = g'$.
\end{lemma}

\begin{proof}
Assume there exists some $b\in\FF_2^m$ such that $f'(b)\neq g'(b)$. Define $b'\in\FF_2^m$ satisfying $b'_j = b_j$ for $j\neq t$ and $b'_t = h(b)$. Then $f(b') = f'(b)\neq g'(b) = g(b')$, which is a contradiction.
\end{proof}

\begin{lemma}\label{lemma2}
Any \added{non-constant} polynomial $f$ \added{with linear factors} can be expressed as\deleted{($k$ coulb be 0)}:
$$
f = g_1 \cdots g_k g.
$$
satisfying

a) $g_1, \cdots, g_k$ are linearly independent linear polynomials and $g$  has no linear factors;

b) $F(g_1)<\cdots<F(g_k)$, and for all $1\leq i\leq k$, $x_{F(g_i)}$ only appears in $g_i$.

\added{If $f$ has no linear factors, we simply write $f=g$.}

We call $f = g_1\cdots g_k g$ is expressed in the restricted form. Furthermore, the restricted form is unique, that is, if $f$ can be expressed in two different ways:
$$
f =g_1 \cdots g_k \cdot g = h_1 \cdots h_j \cdot h
$$
and both of the two forms satisfy the above conditions a) and b), then $k=j$, $g_i =h_i$ for all $1\leq i\leq k$ and $g=h$.
\end{lemma}

The proof of the lemma is in Appendix A.

\begin{remark}
To simplify the enumeration process for the type-I and type-II polynomials, we apply Lemma \ref{lemma2} to reduce the number of different forms. For example, consider the polynomial $f = (x_2+x_1)x_1(x_3x_4+x_5x_6)$. Although it is not in the restricted form, we can rewrite it as $f = x_1(x_2+1)(x_3x_4+x_5x_6)$. This allows us to only consider the restricted forms in the enumeration procedure. However, it is worth noting that it is still challenging to address the third point in Example \ref{ex:0}. For example, we can also write $f$ as $f = x_1(x_2+1)(x_3(x_4+x_5)+x_5(x_6+x_3))$ which still satisfies the conditions in Lemma \ref{lemma2}. Thus, efficient methods need to be explored to eliminate duplicate counting according to the characteristics of the type-I and the type-II polynomials respectively.

\end{remark}

By Lemma \ref{lemma2}, we calculate the number of minimum-weight codewords in decreasing monomial codes, which has been shown in \cite{Bardet2016}. \added{In fact, from Theorem \ref{thm1}, the minimum-weight codewords are exactly the type-I polynomials with $\mu=1$.}

\begin{corollary} \label{cor1}

Let $C(\MI)$ be an $r$-th order decreasing monomial code. Define $\lambda_{x_{i_1}\cdots x_{i_r}} = \sum_{t=1}^r (i_t-t+1)$, then the number of $g_1\cdots g_r$ with $F(\replaced{g_t}{f_t})=i_t$ for $1\leq t\leq r$ is $2^{\lambda_{x_{i_1}\cdots x_{i_r}}}$. Furthermore, \added{the number of minimum-weight codewords is}
\begin{equation}\label{eq:cor1}
A_{2^{m-r}}(C(\MI)) = \sum_{f \in \MI_r} 2^{\lambda_f}.
\end{equation}
where $\MI_r$ is the set of all degree-$r$ monomials in $\MI$.

\end{corollary}

\begin{proof}
According to Theorem \ref{thm1}, each polynomial $f$ with weight $2^{m-r}$ in $C(\MI)$ is a product of $r$ linearly independent polynomials. Then $f$ can be expressed as the restricted form $f = g_1\cdots g_r$. The number of undetermined coefficients in $g_t$ is $(F(g_t)-t+1)$ since the coefficient of $x_{F(g_j)}$ \added{in $g_t$} must be zero for $j<t$. Therefore, the number of polynomials $g_1,\cdots, g_r$ with the largest terms $x_{F(g_1)},\cdots, x_{F(g_r)}$ is $2^{\lambda_{x_{F(g_1)}\cdots x_{F(g_r)}}}$. 

To prove Equation (\ref{eq:cor1}), we verify the three points given in Example \ref{ex:0}. First, Lemma \ref{lemma2} ensures that $g_1,\cdots,g_k$ are linearly independent. Next, since $C(\MI)$ is decreasing, $f\in C(\MI)$ if and only if $x_{F(g_1)}\cdots x_{F(g_r)}\in \MI$. For the third point, Lemma \ref{lemma2} shows that $g_1 \cdots g_r = h_1 \cdots h_r$ if and only if $g_i =h_i$ for all $1\leq t\leq r$, so different forms lead to different polynomials in this situation.

Therefore, $A_{2^{m-r}}(C(\MI))$ is equal to the sum of $2^{\lambda_f}$ where $f \in \MI_r$.

\end{proof}

\begin{example}\label{ex:1}
\replaced{Let $C(\MI)$ be the polar code defined in Example \ref{ex:0}}{. Let $C(\MI)$ be a [128,80] decreasing polar code where $\MI$ is generated by $\MI_{\min}=\{x_1x_2x_5x_7,x_1x_3x_4x_7,x_1x_4x_5x_6,x_2x_3x_5x_6,x_4x_6x_7\}$.} with minimum weight 8.

For example, if $f = x_1x_2x_5x_7$ then
$$
\lambda_f = 1+(2-1)+(5-2)+(7-3) = 9, 
$$
so the number of  $g_1\cdots g_4$ with $F(g_1) = 1, F(g_2) = 2, F(g_3) = 5, F(g_4) = 7$ is $2^9 = 512$.

We calculate that $A_8(C(\MI)) = \sum_{f \in \MI_4} 2^{\lambda_f} = 5680$.
\end{example}

Next, we extend our analysis to other low-weight codewords. Since each low-weight codeword can be expressed as either a type-I or type-II polynomial, we calculate the number of both types respectively.

\subsection{Number of the type-I polynomials}\label{ss32}

In this subsection, our objective is to calculate the number of polynomials which can be expressed as equation (\ref{eq1}) in the $r$-$th$ order decreasing monomial code $C(\MI)$ when $\mu\geq 2$. The case $\mu = 1$\added{, i.e., the number of minimum-weight codewords,} has already been solved in Corollary \ref{cor1}. Thanks to Theorem \ref{thm1}, we can classify the type-I polynomials into disjoint set and the number of polynomials in each set is calculated in Theorem \ref{thm:m1}.

For a type-I polynomial $f$, we define the form $f = P(g_1, \cdots, g_{r+2 \mu-2})$ to be proper if both $f$ itself and $g_{r-3+2t}g_{r-2+2t}$ for each $1\leq t \leq\mu$ satisfy the two conditions a) and b) in Lemma \ref{lemma2}. According to Lemma \ref{lemma2}, each $f$ can \replaced{be}{been} expressed as a proper form, so we only need to consider proper forms from now on.

\begin{example}
$f = x_1((x_4+x_2)(x_2+1) + x_3(x_5+x_1))$ is not proper since the largest term of $x_2+1$ is $x_2$, which appears in $x_4+x_2$. We can rewrite it as the proper form $f = x_1(x_2(x_4+1) + x_3(x_5+1))$.
\end{example}

Define $A_{\mu}$ to be the set of proper polynomials $P(g_1, \cdots, g_{r+2 \mu-2}) \in C(\MI)$.  Denote $u = (i_1,\cdots,i_{r+2 \mu-2})$, and define $A_u$ to be the subset of $A_{\mu}$ such that $F(g_t) = i_t$. Define $S_{\mu} = \{(i_1,\cdots,i_{r+2 \mu-2}) \mid P(x_{i_1}, \cdots, x_{i_{r+2\mu-2}}) \in C(\MI), i_{r-1}<i_{r+1}<\cdots<i_{r+2 \mu-3}, \text { and } i_1,\cdots,i_{r+2 \mu-2} \text{ are all different} \}$. Due to the proper condition, we also have $i_1<i_2<\cdots< i_{r-2}$ and $i_{j}<i_{j+1}$ for $j=r-1,r+1,\cdots,r+2 \mu-3$. 

In fact, $S_{\mu}$ is defined to avoid the repetition of counting due to exchange of polynomials (e.x.  $g_1 \cdots g_{r-2}\left(g_{r-1} g_r+g_{r+1} g_{r+2}+\cdots+g_{r+2 \mu-3} g_{r+2 \mu-2}\right) = g_1 \cdots g_{r-2}\left(g_{r+1} g_{r+2}+g_{r-1} g_r+\cdots+g_{r+2 \mu-3} g_{r+2 \mu-2}\right)$). By defining $S_{\mu}$, we ensure each type-I polynomial belongs to a unique $A_u$ for some $u\in S_{\mu}$. This prevents duplication in the classification of the polynomials and allows for a clear and distinct categorization.

\begin{theorem}\label{thm2}
$$
A_{\mu} = \bigcup_{u\in S_{\mu}} A_u.
$$
And different $A_u$ are disjoint. Therefore,
$$
|A_{\mu}| = \sum_{u\in S_{\mu}} |A_u|.
$$
\end{theorem}

In order to prove the theorem \ref{thm2}, we rely on the following lemmas to show that every polynomial $f=\sum_{i=1}^{\mu} f_{2 i-1} f_{2 i}$ can be rewritten so that the largest terms of linear polynomials $f_i$ are different regardless of the order.

\begin{lemma}\label{lemma3}
Let $f_1,...,f_{2\mu}$ be linearly independent polynomials and $f=\sum_{i=1}^{\mu} f_{2 i-1} f_{2 i}$. Then there exists linearly independent polynomials $g_1,...,g_{2\mu}$ with different largest terms such that $f=\sum_{i=1}^\mu g_{2 i-1} g_{2 i}$.
\end{lemma}

The proof of the lemma is in Appendix B.

\begin{lemma}\label{lemma4}
Assume $\sum_{i=1}^{\mu} f_{2 i-1} f_{2 i} = \sum_{i=1}^\mu g_{2 i-1} g_{2 i}$, where all $F(f_i)$ are different, $F(f_1)<F(f_3)<\cdots<F(f_{2\mu-1})$, $F(f_{2 j-1}) <  F(f_{2 j})$ for $1\leq j\leq \mu$, and $g_i$ satisfy similar conditions. Then $F(f_i) = F(g_i)$ for all $1\leq i \leq 2\mu$.
\end{lemma}

The proof of the lemma is in Appendix C.

\begin{proof}[Proof of Theorem \ref{thm2}]

Denote $f = P(f_1, \cdots,  f_{r+2 \mu-2})\in A_{\mu}$. Since $f$ is proper, for all $1\leq i\leq r-2$, $x_{F(f_i)}$ only appears in $f_i$. According to Lemma \ref{lemma3}, we suppose that $F(f_1),\cdots,F(f_{r+2\mu-2})$ are all different. Therefore, each polynomial in $A_{\mu}$ belongs to some $A_u$ after sorting the linear polynomials based on the largest terms.

Now let us assume there exists another proper form $g = P(g_1, \cdots, g_{r+2 \mu-2})\in A_v$ such that $f=g$. By Lemma \ref{lemma2}, we find that $f_i=g_i$ for $1\leq i\leq r-2$ and $f_{r-1} f_r+\cdots+f_{r+2 \mu-3} f_{r+2 \mu-2} = g_{r-1} g_r+\cdots+g_{r+2 \mu-3} g_{r+2 \mu-2}$. Since $u,v\in S_{\mu}$, the conditions in Lemma \ref{lemma4} are satisfied. Thus, $F(f_i)= F(g_i)$ for $r-1\leq i\leq r+2 \mu-2$. Therefore, we conclude that $u=v$, which implies that the different sets $A_u$ are disjoint.

\end{proof}

Now we proceed to calculate $|A_{\mu}|$. For $u = (i_1,\cdots,i_{r+2 \mu-2})\in S_{\mu}$. We aim to determine the number of different polynomials which can be expressed as $f = P(f_1,\cdots, f_{r-2}, h_1,g_1, \cdots, h_{\mu}, g_{\mu}) = f_1 \cdots f_{r-2}\left(h_1 g_1+\cdots+h_{\mu} g_{\mu}\right) \in A_u$ by assigning values to the coefficients in $f_1,\cdots, f_{r-2}, h_1,g_1, \cdots, h_{\mu}, g_{\mu}$. 

We verify the three points in Example \ref{ex:0}. Firstly, since $i_1,\cdots,i_{r+2 \mu-2}$ are all different, the linearly independence is naturally satisfied. Next, since $C(\MI)$ is decreasing, $f\in C(\MI)$ if and only if $P(x_{i_1},\cdots, x_{i_{r+2 \mu-2}})\in C(\MI)$, i.e.,  the $\mu$ polynomials $x_{i_1}\cdots x_{i_r}, x_{i_1}\cdots x_{i_{r-2}} x_{i_{r+1}}x_{i_{r+2}},\cdots,$ \\
$x_{i_1}\cdots x_{i_{r-2}} x_{i_{r+2\mu-3}}x_{i_{r+2\mu-2}} \in \MI$. However, the challenge arises from the fact that different forms may lead to the same polynomial. 

Since $f$ is proper, different $f_t$ with $1\leq t\leq r-2$ will result in different $f$ by Lemma \ref{lemma2}. Thus, there are $ 2^{\lambda_{x_{i_1}\cdots x_{i_{r-2}}}}$ choices for $f_1,\cdots, f_{r-2}$. 

For $u = (i_1,\cdots,i_{r+2 \mu-2})\in S_{\mu}$, define $\alpha_u = |\{(k,t)\mid 1\leq k<t\leq \mu, i_{r-2+2k}>i_{r-2+2t}> i_{r-3+2k}>i_{r-3+2t} \}|$, \replaced{$\beta_u = \sum_{t=1}^{\mu}(\bar{\varphi}_{u, [1,r+2\mu-2]}(i_{r-2+2t}) - \bar{\varphi}_{u, [1,r+2\mu-2]}(i_{r-3+2t}))$}{$\beta_u = |\{(s,k)\mid 1\leq s\leq m, 1\leq k\leq \mu, s\neq i_j, i_{r-2+2k}>s> $ $i_{r-3+2k} \}|$}, \added{i.e., the number of the integer pairs $(s,t)$ satisfying $i_{r-2+2t}>s> i_{r-3+2t}$ and $s$ does not appear in $u$,} \replaced{$\gamma_u = \sum_{t=1}^{\mu} \bar{\varphi}_{u,[1, r-2]}(i_{r-3+2t}),$}{$\gamma_u = \sum_{t=1}^{\mu} (i_{r-3+2t} - |\{1\leq k\leq r-2 \mid i_k< i_{r-3+2t}\}|)$} \added{i.e, the number of the integer pairs $(s,t)$ satisfying $i_{r-3+2t}>s$ and $s$ does not appear in $i_1,\cdots,i_{r-2}$}. In fact, $\alpha_u + \beta_u + 2\gamma_u$ is the number of free coefficients in $\sum_{j=1}^{\mu} h_jg_j$ which can be determined arbitrarily.

\begin{theorem} \label{thm:m1}
For $u = (i_1,\cdots,i_{r+2 \mu-2})\in S_{\mu}$,
\begin{equation} \label{eq4}
|A_u| = 2^{\lambda_{x_{i_1}\cdots x_{i_{r-2}}}+\alpha_u + \beta_u+ 2\gamma_u}.
\end{equation}
Therefore,
\begin{equation} \label{eq5}
|A_{\mu}| = \sum_{u\in S_{\mu}} |A_u| =  \sum_{u\in S_{\mu}}2^{\lambda_{x_{i_1}\cdots x_{i_{r-2}}}+\alpha_u + \beta_u+ 2\gamma_u}.
\end{equation}

\end{theorem}

The proof of the theorem is in Appendix D.

The Algorithm \ref{alg:1} concludes the procedure for calculating $|A_{\mu}|$.

\begin{figure}[!t]
\begin{algorithm}[H]
\caption{Calculate $|A_{\mu}|$}
\begin{algorithmic}[1]\label{alg:1}

\renewcommand{\algorithmicrequire}{\textbf{Input:}}
\renewcommand{\algorithmicensure}{\textbf{Output:}}
\REQUIRE information set $\MI$ of the $r$-order decreasing polar code, weight $2^{m-r+1}-2^{m-r+1-\mu}$.
\ENSURE the number $M$ of the type-I polynomials with weight $2^{m-r+1}-2^{m-r+1-\mu}$.
\STATE $M \gets 0$
\FOR{ $u = (i_1,\cdots,i_{r+2 \mu-2}) \in S_{\mu}$} 
 \STATE Calculate the number $\alpha_u$ of the integer pairs $(k,t)$ satisfying $1\leq k<t\leq \mu, i_{r-2+2k}>i_{r-2+2t}> i_{r-3+2k}>i_{r-3+2t}$;
 \STATE Calculate the number $\beta_u$ of the \replaced{integer pairs $(s,t)$}{integer $s$} satisfying $i_{r-2+2t}>s> i_{r-3+2t}$ and $s$ does not appear in $u$;
 \STATE  Calculate the number $\gamma_u$  \replaced{of the integer pairs $(s,t)$ satisfying $i_{r-3+2t}>s$ and $s$ does not appear in $i_1,\cdots,i_{r-2}$}{$\gets \sum_{t=1}^{\mu} (i_{r-3+2t} - |\{1\leq k\leq r-2 \mid i_k< i_{r-3+2t}\}|)$};
 \STATE $M\gets M+2^{\lambda_{x_{i_1}\cdots x_{i_{r-2}}}+\alpha_u + \beta_u+ 2\gamma_u}$;
\ENDFOR

\end{algorithmic}
\end{algorithm}
\end{figure}

\begin{remark}
As $m$ approaches infinity, the complexity of calculating $\lambda_{x_{i_1}\cdots x_{i_{r+2\mu-2}}}, \alpha_u, \beta_u$ and $\gamma_u$ is \replaced{$O(r+2\mu-2)$, $O(\mu^2)$, $O(\mu m)$ and  $O(\mu r)$}{$O(m^2)$ since the length of $u$ is $r+2\mu-2\leq m$} \added{respectively}. Consequently, the complexity of calculating $|A_u|$ is $O( \replaced{\mu m}{m^2})$. Additionally, the size of $S_{\mu}$ is smaller \replaced{than}{that} $\binom{m}{r+2\mu-2}$. Therefore, the time complexity of Algorithm \ref{alg:1} is $O( \replaced{\mu m}{m^2}\binom{m}{r+2\mu-2})$.
\end{remark}

\begin{remark}
\added{When $\mu = 2$, Theorem \ref{thm:m1} provides the same result as Theorem 4 in \cite{Rowshan2023} with the number of sets $S_2$ being identical. Moreover, the complexity associated with calculating the size of each set is also at $O(m)$, thus the complexity is similar.}
\end{remark}

\begin{example}\label{ex:2}
\replaced{Let $C(\MI)$ be the polar code defined in Example \ref{ex:0}}{. Let $C(\MI)$ be a [128,80] decreasing polar code where $\MI$ is generated by $\MI_{\min}=\{x_1x_2x_5x_7,x_1x_3x_4x_7,x_1x_4x_5x_6,x_2x_3x_5x_6,x_4x_6x_7\}$.} We have calculated $A_8((C(\MI))$ in Example \ref{ex:1}. Now we calculate $A_{12}(C(\MI))$ with $\mu = 2$.

For example, if $u = (1,7,2,5,3,4)\in S_2$, $|A_u|$ is the number of different polynomials $f_1f_2(h_1g_1+ h_2g_2) $ with $F(f_1) = 1, F(f_2) = 7, F(h_1) = 2, F(g_1) = 5, F(h_2) = 3$ and $F(g_2) = 4$. Then $\lambda_{x_1x_7} = 7$. Since $F(g_1)>F(g_2)>F(h_2)>F(h_1)$, we have $\alpha_u = 0$. Besides, \added{$\beta_u$ is the sum of number of $s$ does not appear in $u$ satisfying $5>s>2$ and $4>s>3$. Since only 6 does not appear in $u$, $\beta_u = 0$.}\deleted{we calculate $\beta_u = 0$ and} \added{Next, $\varphi_{u,[1, 2]}(2) = \varphi_{u,[1, 2]}(3) = 1$, so} $\gamma_u = 2 - 1 + 3 - 1 = 3$. In fact, $x_7, x_6$ and $x_1$ do not appear in $g_1h_1+g_2h_2$. For $x_5$, the coefficients of $x_5$ can be determined arbitrarily by assigning $h_1$. For $x_4$, the coefficients of $x_4x_2$ and $x_4$ can be determined arbitrarily by assigning $h_2$. For $x_3$ and $x_2$, the coefficients of $x_3x_2$, $x_3$ and $x_2$ can be determined arbitrarily. Therefore, $|A_u| = 2^{7+6} = 8192$. 

From Theorem \ref{thm1} and \ref{thm:m1}, we have $A_{12}(C(\MI)) = \sum_{u\in S_2} |A_u| =  \sum_{u\in S_2}2^{\lambda_{x_{i_1}\cdots x_{i_{r-2}}}+\alpha_u + \beta_u+ 2\gamma_u} = 508672$.
\end{example}

\subsection{Number of the type-II polynomials}\label{ss33}

In this subsection,  our objective is to calculate the number of polynomials which can be expressed as equation (\ref{eq2}) for $\mu\geq 3$ in the $r$-$th$ order decreasing monomial code $C(\MI)$.  As in section \ref{ss32}, we classify the type-II polynomials into disjoint set and the number of polynomials in each set is calculated in Theorem \ref{thm:m2}.

For a type-II polynomial $f$, we define the form $f = Q(g_1, \cdots, g_{r+\mu})$ to be proper if $f$ itself, $g_{r-\mu+1}\cdots g_{r}$ and $g_{r+1}\cdots g_{r+ \mu}$ satisfies the conditions a) and b) in Lemma \ref{lemma2},  also satisfy those conditions. According to Lemma \ref{lemma2}, each $f$ can been expressed as a proper form, so we only need to consider proper forms from now on.

Let $a=(a_1,...,a_{\mu}), b=(b_1,...,b_{\mu})\in \mathbb{N}^{\mu}$, define the lexicographic order $a\leq_l b$ if and only if there exists some $i$ such that $a_i<b_i$ and $a_j = b_j$ for all $i<j\leq \mu$ or $a_j = b_j$ for all $1\leq j\leq \mu$.

Define $B_{\mu}$ to be the set of polynomials $Q(g_1,\cdots, g_{r+\mu})\in C(\MI)$. Denote $u = (i_1,\cdots,i_{r+\mu})$. Define $B_u$ to be the subset of $B_{\mu}$ such that $F(g_t) = i_t$. Define $T_{\mu} = \{(i_1,\cdots,i_{r+\mu}) \mid Q(x_{i_1},\cdots, x_{i_{r+\mu}}) \in C(\MI), (i_{r-\mu+1}, \cdots,i_r) \leq_l (i_{r+1},\cdots i_{r+\mu}) \}$. Due to the proper condition, we also have $i_1<i_2<\cdots< i_{r-\mu}, i_{r-\mu+1}<\cdots<i_r, i_{r+1}<\cdots<i_{r+\mu}$, and $i_1,\cdots,i_r$ are all distinct as well as $i_1,\cdots,i_{r-\mu+1}, i_{r+1},\cdots, i_{r+\mu}$.

In fact, the introduction of the lexicographic order in $T_{\mu}$ is to avoid repetition of counting due to exchange of polynomials (e.g., $g_1 \cdots g_{r-\mu}\left(g_{r-\mu+1} \cdots g_r+g_{r+1} \cdots g_{r+\mu}\right) = g_1 \cdots g_{r-\mu}\left(g_{r+1} \cdots g_{r+\mu}+g_{r-\mu+1} \cdots g_r\right)$). 

By defining $T_{\mu}$, we ensure each type-II polynomial belongs to a unique $B_u$ for some $u\in S_{\mu}$. This prevents duplication in the classification of the polynomials and allows for a clear and distinct categorization, similar to what we have seen in section \ref{ss32}.

\begin{theorem}\label{thm4}
$$
B_{\mu} = \bigcup_{u\in T_{\mu}} B_u.
$$
And different $B_u$ are disjoint. Therefore,
$$
|B_{\mu}| = \sum_{u\in T_{\mu}} |B_u|.
$$
\end{theorem}

In order to prove the theorem \ref{thm4}, we rely on the following lemma presented in \cite{Kasami1970} (L6) to show that different $g_{r-\mu+1}, \cdots, g_{r+\mu}$ lead to different polynomials.

\begin{lemma}[\cite{Kasami1970}] \label{lemma5}

Suppose $\mu \geq 3$, $f=g+h=g^{\prime}+h^{\prime}$, $g=y_1 \cdots y_\mu$, $h=y_{\mu+1} \cdots y_{2 \mu}$, $g^{\prime}=y_1^{\prime} \cdots y_\mu^{\prime}$, $h^{\prime}=y_{\mu+1}^{\prime} \cdots y_{2 \mu}^{\prime}$, $y_1, \cdots, y_{2 \mu}$ are linearly independent, $y_1^{\prime}, \cdots, y_{2 \mu}^{\prime}$ are also linearly independent. Then
$$
g=g^{\prime}, \quad h=h^{\prime} \quad \text { or } \quad g=h^{\prime}, \quad h=g^{\prime}.
$$

\end{lemma}

\begin{proof}[Proof of Theorem \ref{thm4}]

Let $f = Q(f_1, \cdots,  f_{r+\mu})\in B_{\mu}$. Since $f$ is proper, every polynomial in $B_{\mu}$ belongs to some $B_u$ after sorting the linear polynomials based on their largest terms.

If there exists another proper form $g = Q(g_1, \cdots, g_{r+\mu})\in B_v$ such that $f=g$. According to Lemma \ref{lemma2}, we have $f_i=g_i$ for $1\leq i\leq r-\mu$ and $f_{r-\mu\added{+1}} \cdots f_r + f_{r+1} \cdots f_{r+\mu} = g_{r-\mu\added{+1}} \cdots g_r + g_{r+1} \cdots g_{r+\mu}$. Then by Lemma \ref{lemma5} and the lexicographic order, we deduce that $f_{r-\mu\added{+1}} \cdots f_r = g_{r-\mu\added{+1}} \cdots g_r,  f_{r+1} \cdots f_{r+\mu} = g_{r+1} \cdots g_{r+\mu}$. Once again, applying Lemma \ref{lemma2}, $f_i=g_i$ for $r-\mu + 1\leq i\leq r+\mu$. Therefore, we conclude that $u=v$, which implies different sets are disjoint.

\end{proof}

Now we proceed to calculate $|B_{\mu}|$. Denote $u = (i_1,\cdots,i_{r+\mu})\in T_{\mu}$. From the proof of Theorem \ref{thm4}, $Q(f_1 \cdots  f_{r+\mu}) = Q(g_1 \cdots g_{r+\mu})$ if and only if $f_i=g_i$ for all $1\leq i\leq r+\mu$, so the third point in Example \ref{ex:0} is satisfied trivially. However, we still need to check the first two points.

There are $\lambda_{x_{i_1}\cdots x_{i_r}}$ choices for $f_1,\cdots, f_r$. Next, we need to choose linearly independent $f_{r+1},\cdots, f_{r+\mu}$ with given largest terms such that $Q(f_1 \cdots  f_{r+\mu})\in C(\MI)$. 

Now, if $i_{r-\mu+j} \neq i_{r+j}$ for some $1\leq j\leq \mu$, $Q(x_{i_1} \cdots  x_{i_{r+\mu}})\in C(\MI)$ if and only if $Q(f_1 \cdots  f_{r+\mu}) \in C(\MI)$. However, if $i_{r-\mu+j} = i_{r+j}$ for all $1\leq j\leq \mu$, $Q(x_{i_1} \cdots  x_{i_{r+\mu}}) = 0\in C(\MI)$ does not imply $Q(f_1 \cdots  f_{r+\mu})\in C(\MI)$, so we need to examine the second point carefully. Define $b_u(k) = \max\{t \mid x_{i_1}\cdots x_{i_{r-\mu}}x_{i_{r+1}}\cdots x_{i_{r+k-1}} x_t x_{i_{r+k+1}}\cdots x_{i_{r+\mu}} \in \MI, t\neq i_{r+1},\cdots,i_{r+\mu} \}$. If $ x_{i_1}\cdots x_{i_{r-\mu}}x_{i_{r+1}}\cdots x_{i_{r+k-1}} x_1 x_{i_{r+k+1}}\cdots x_{i_{r+\mu}}\notin \MI$, $b_u(k)$ is defined to be 0. In this situation, $Q(f_1 \cdots  f_{r+\mu})\in C(\MI)$ if and only if for all $1 \leq k \leq \mu$, the coefficients of $x_j$ with $j>b_u(k)$ in $f_{r-\mu+k}$ and $f_{r+k}$ are equal. If $b_u(k) = 0$, then $f_{r-\mu+k} = f_{r+k} + a_0$ which contradicts against the linear independence condition, so $B_u = \varnothing$. 

\deleted{Define the number of indices in $i_1,\cdots,i_{r-\mu}$ less than $k$ to be $\varphi_u(k) = |\{1\leq j\leq r-\mu \mid i_j < k \}|$, and define the number of indices in $i_{r-\mu+1},\cdots,i_r$ less than $k$ to be $p_u(k) = |\{r-\mu+1\leq j\leq r \mid  i_j < k\}|$. }

Define $s_u(k)$ to be the number of choices for $f_{r+k}$.

Case 1) $i_{r-\mu+j} \neq i_{r+j}$ for some $1\leq j\leq \mu$, 

Subcase I) $i_{r+k} \neq i_{r-j}$ for all $0\leq j\leq \mu-1$, $s_u(k) = 2^{\added{\bar{\varphi}_{u,[1,r-\mu]}(i_{r+k})} - k + 1}$.

Subcase II) $i_{r+k} = i_{r-j}$ for some $0\leq j\leq \mu-1$, $s_u(k) = 2^{\added{\bar{\varphi}_{u,[1,r-\mu]}(i_{r+k})} - k + 1}-2^{\added{\varphi_{u,[r-\mu+1,r]}(i_{r+k})}+1}$ if $\added{\bar{\varphi}_{u,[1,r-\mu]}(i_{r+k})} - k > \added{\varphi_{u,[r-\mu+1,r]}(i_{r+k})}$ and 0 otherwise.

Case 2) $i_{r-\mu+j} = i_{r+j}$ for all $1\leq j\leq \mu$. 

$s_u(k) = 2^{\added{\bar{\varphi}_{u,[1,r]}(b_u(k))+1}}-2^k$ if $\added{\bar{\varphi}_{u,[1,r]}(b_u(k))} + 1 > k$ and 0 otherwise.

Finally, define $\sigma_u = \prod_{k=1}^{\mu} s_u(k)$ for Case 1 and $\sigma_u = (\prod_{k=1}^{\mu} s_u(k))/2$ for Case 2. We need to divide the value by 2 in Case 2 since $Q(f_1, \cdots,  f_{r+\mu}) = Q(f_1, \cdots, f_{r-\mu}, f_{r+1}, \cdots, f_{r+\mu}, f_{r-\mu+1},\cdots, f_{r})$ is calculated twice.

Next, we prove the correctness of the definition of $s_u(k)$ and derive $|B_\mu|$ in Theorem 5.

\begin{theorem}\label{thm:m2}

For $u = (i_1,\cdots,i_{r+\mu})\in T_{\mu}$,
\begin{equation} \label{eq6}
|B_u| = 2^{\lambda_{x_{i_1}\cdots x_{i_r}}}\sigma_u.
\end{equation}
Therefore,
\begin{equation} \label{eq7}
|B_{\mu}| = \sum_{u\in T_{\mu}} |B_u| =  \sum_{u\in T_{\mu}} 2^{\lambda_{x_{i_1}\cdots x_{i_r}}}\sigma_u.
\end{equation}

\end{theorem}

The proof of the theorem is in Appendix E.

The Algorithm \ref{alg:2} concludes the procedure for calculating $|B_{\mu}|$.

\begin{figure}[!t]
\begin{algorithm}[H]
\caption{Calculate $|B_{\mu}|$}
\begin{algorithmic}[1]\label{alg:2}

\renewcommand{\algorithmicrequire}{\textbf{Input:}}
\renewcommand{\algorithmicensure}{\textbf{Output:}}
\REQUIRE the information set $\MI$ of $r$-order decreasing polar codes, weight $2^{m-r}-2^{m-r-\mu}$.
\ENSURE  the number $M$ of type-II polynomials with weight $2^{m-r+1}-2^{m-r+1-\mu}$.
\STATE $M \gets 0$;
\FOR{$u = (i_1,\cdots,i_{r+\mu})\in T_{\mu}$} 
 \IF{there exists some $1\leq j\leq \mu$ such that $i_{r-\mu+j} \neq i_{r+j}$}
 \STATE $L\gets 1$;
  \FOR{$k=1$ to $\mu$}
   \IF{there exists some $0\leq j\leq \mu-1$ such that $i_{r+k} = i_{r-j}$}
    \STATE $L\gets 2^{\added{\bar{\varphi}_{u,[1,r-\mu]}(i_{r+k})}- k + 1} L$;
   \ELSE
    \IF{$\added{\bar{\varphi}_{u,[1,r-\mu]}(i_{r+k})}- k\leq \added{\varphi_{u,[r-\mu+1,r]}(i_{r+k})}$}
     \STATE $L\gets 0$;
     \STATE {\bf Return}
    \ELSE
     \STATE $L\gets 2^{(\added{\bar{\varphi}_{u,[1,r-\mu]}(i_{r+k})} - k + 1}-2^{\added{\varphi_{u,[r-\mu+1,r]}(i_{r+k})}+1})L$; 
    \ENDIF
   \ENDIF
  \ENDFOR
 \STATE $M\gets M+2^{\lambda_{x_1\cdots x_{i_r}}} L$;
 
 \ELSE
  \STATE $L\gets 1$;
  \FOR{$k=1$ to $\mu$}
    \IF{$\added{\bar{\varphi}_{u,[1,r]}(b_u(k))} + 1\leq k$}
     \STATE $L\gets 0$;
     \STATE {\bf Return}
    \ELSE
     \STATE $L\gets (2^{\added{\bar{\varphi}_{u,[1,r]}(b_u(k))}+1}-2^k)L$; 
    \ENDIF
  \ENDFOR
  \STATE $M\gets M+2^{\lambda_{x_1\cdots x_{i_r}}} L/2$;
 \ENDIF
\ENDFOR

\end{algorithmic}
\end{algorithm}
\end{figure}

\begin{remark}
As $m$ approaches infinity, the complexity of calculating $\added{\varphi_{u,[1,r-\mu]}(k), \bar{\varphi}_{u,[r-\mu+1,r]}(k)}$ and $b_u(k)$ is $O(m)$ since $u$ is vector with length $r+\mu\leq m$. So the complexity of calculating $s_u(k)$ is $O(m)$. Then the complexity of calculating $\sigma_u$ and further $|B_u|$ is $O(m^2)$. The size of $T_{\mu}$ is $\binom{m}{r}\binom{m-r+\mu}{\mu}$. Therefore, the time complexity of Algorithm \ref{alg:2} is $O(m^2\binom{m}{r}\binom{m-r+\mu}{\mu})$.
\end{remark}

\begin{example}
\replaced{Let $C(\MI)$ be the polar code defined in Example \ref{ex:0}.}{Let $C(\MI)$ be a [128,80] decreasing polar code where $\MI$ is generated by $\MI_{\min}=\{x_1x_2x_5x_7,x_1x_3x_4x_7,x_1x_4x_5x_6,x_2x_3x_5x_6,x_4x_6x_7\}$.} We have calculated $A_8((C(\MI))$ and $A_{12}((C(\MI))$ in Example \ref{ex:1} and \ref{ex:2}. Now we calculate $A_{14}(C(\MI))$ with $\mu = 3$.

For example, if $u = (1,2,5,7,3,4,7)\in T_3$ then $\lambda_{x_1x_2x_5x_7} = 9$. For the choice of $f_5$, we have $F(f_5) = 3$, $x_3$ does not appear in $x_2x_5x_7$ and $\added{\bar{\varphi}_{u,[1,1]}(3)}=2$, so there are $2^{2-1+1}=4$ choices for $f_5$. Similarly there are 4 choices for $f_6$. For $f_7$, we have $F(f_7) = 7$ and $x_7$ appears in $x_2x_5x_7$, since $\added{\varphi_{u,[2,4]}}(7)=2$, there are $2^{6-3+1} - 2^{2+1} = 8$ choices for $f_7$. Therefore, $|B_u| = 2^9\times 4\times 4\times 8 = 65536$.

In another example, if $v = (1,4,5,7,4,5,7)\in T_3$ then $\lambda_{x_1x_4x_5x_7} = 11$. Since $x_1x_3x_5x_7\notin\MI$ and $x_1x_2x_5x_7\in\MI$, $b_v(1) = 2, \added{\bar{\varphi}_{v,[2,4]}(b_v(1))}=2$. Similarly, $b_v(2) = 3, \added{\bar{\varphi}_{v,[2,4]}(b_v(2))} =3,  b_v(3) = 6, \added{\bar{\varphi}_{v,[2,4]}(b_v(3))} = 4$. Hence, $f_7, f_6, f_5$ have $2^2-2 = 2, 2^3-2^2 = 4, 2^4 - 2^3 = 16$ choices, respectively. Therefore, $|B_v| = 2^{11}\times 2\times 4\times 16 /2= 131072$.

By Theorem \ref{thm1} and \ref{thm:m2}, it can be calculated that $A_{14}(C(\MI)) = \sum_{v\in T_{\mu}}2^{\lambda_{x_{v_1}}\cdots x_{v_{r-\mu}}}\sigma_v = 1835008$.
\end{example}

\section{Number of Low-weight polynomials}

In this section, we provide the closed-form expressions and the enumeration complexity for the number of codewords with weight less than $2w_{\min}$.

By Corollary \ref{cor1}, Theorem \ref{thm:m1} and Theorem \ref{thm:m2}, we have

\begin{theorem}\label{thm:me}
Let $(C(\MI))$ be an $r$-order monomial code and $\mu$ be a positive integer. Denote $M = A_{2^{m-r+1}-2^{m-r+1-\mu}}(C(\MI))$.

When $\mu = 1$,
\begin{equation}
M = \sum_{f \in \MI_r} 2^{\lambda_f}.
\end{equation}

When $\mu = 2$,
\begin{equation}
M = \sum_{u\in S_{\mu}}2^{\lambda_{x_{i_1}\cdots x_{i_{r-2}}}+\alpha_u + \beta_u+ 2\gamma_u}.
\end{equation}

When $3\leq \mu \leq \min\{(m-r+2)/2,m-r,r\}$,
\begin{equation}
M = \sum_{u\in T_{\mu}} 2^{\lambda_{x_{i_1}\cdots x_{i_{r-2}}}}\sigma_u +\sum_{u\in S_{\mu}}2^{\lambda_{x_{i_1}\cdots x_{i_{r-\mu}}}+\alpha_u + \beta_u+ 2\gamma_u}.
\end{equation}

\end{theorem}

The closed-form expressions for case $\mu = 1$ and 2 have been discovered in \cite{Bardet2016} and \cite{Rowshan2023}.

\begin{remark}
As $m$ approaches infinity, the complexity of calculating the number of codewords with weight less than 2$w_{\min}$ in a decreasing monomial code is determined by the combined complexity of Algorithm \ref{alg:1} and \ref{alg:2}, which is given by $O(m^2(\binom{m}{r+2\mu-2}+\binom{m}{r}\binom{m-r+\mu}{\mu}))$. 

When $r$ is a constant, i.e., independent of $m$, the time complexity is $O(m^{2+\max\{r+2\mu-2, r+\mu\}}) = O((\log N)^{2+\max\{r+2\mu-2, r+\mu\}})$ where $N$ is code length. 

When $r$ is proportional to $m$, i.e., $r = am, a\in (0,1)$, the time complexity is $O(m^22^{\max\{1,2h(a)\}m}) = O((\log N)^2N^{\max\{1,2h(a)\}})$ where $h(a) = -a\log a - (1-a)\log(1-a)$.
\end{remark}

Therefore, our algorithms provide polynomial time complexity for calculating the number of codewords with weight less than $2w_{\min}$.  In comparison to the approach of using SCL decoding with a very large list to collect low-weight codewords \cite{Li2012}, our algorithms are efficient and guarantee an accurate collection of low-weight codewords.

\begin{example}
\added{We check our results for RM$(6,3)$. From Theorem \ref{thm:me}, we calculate that $A_{8}(C(\MI)) = 11160$, $A_{12}(C(\MI)) =  174999$ and $A_{14}(C(\MI))= 28555680$, which is consistent with the results obtained from Theorem \ref{thm1} \cite{Kasami1970}. Therefore, our result is correct for RM codes.}

\added{Next we examine our findings for $C(\MI)$ where $\MI$ is generated by $\MI_{\min}=\{15 ,21, 55\}$. Since $C(\MI)$ is a subcode of RM$(6,3)$, we list all the polynomials of RM$(6,3)$ with weight less than 16 and verify whether each one is a codeword in $C(\MI)$. The result shows that $A_{8}(C(\MI)) = 2456$, $A_{12}(C(\MI)) =  142208$ and $A_{14}(C(\MI))= 868325$, which yields the same value as Theorem \ref{thm:me}.}

\added{We also validate our results for the code whose weight distribution was previously calculated in \cite[Table VIII]{Yao2023}. By calculating the codewords with less than $16$ in $C(\MI)$, we obtain $A_{8}(C(\MI)) = 304$, $A_{12}(C(\MI)) =  768$ and $A_{14}(C(\MI))= 0$. These results perfectly align with those presented in \cite[Table VIII]{Yao2023}.  It is worth noting that the complexity of the method in \cite{Yao2023} is exponential, while our algorithm runs in polynomial time relative to the code length.}

\end{example}

\begin{example}\label{Ex11}

Table \ref{tab1} \added{and Table \ref{tab2}} present a comparison of the weight distribution of different [128,80] \added{and [256, 94]} decreasing monomial codes with minimum distance 8 \added{and 16, respectively}. It is worth noting that different constructions yield significantly different weight distributions.

\begin{table}[htbp]
\begin{center} 
\setlength{\tabcolsep}{1.5pt}{
\begin{tabular}{c|c|c|c|c} 
\hline
& $\MI_{\text{min}}$	& $A_{8}$ & $A_{12}$ & $A_{14}$
 \\
\hline
code 1 & $\{23, 44, 50, 70, 73\}$ & 5680 & 508672 & 1835008\\

code 2 & $\{15, 28, 73\}$ & 5168 & 367360 & 1376256\\

code 3 & $\{23, 38, 97\}$ & 7216 & 596736 & 4128768\\

code 4 & $\{29, 39, 41\}$  & 8752 & 897792 & 6881280\\
\hline
\end{tabular}}
\caption{The number of codewords in [128,80] codes with weight less than 16}
\label{tab1}
\end{center}
\end{table}

\begin{table}[htbp] 
\begin{center} 
\setlength{\tabcolsep}{1.5pt}{
\begin{tabular}{c|c|c|c|c} 
\hline
& $\MI_{\text{min}}$	& $A_{16}$ & $A_{24}$ & $A_{28}$
 \\
\hline

code 1 & $\{63, 91, 103, 120, 204, 210, 225\}$ & 1584 & 49920 & 0\\

code 2 & $\{63, 94, 103, 151, 180, 204, 209\}$ & 2096 & 84376 & 0\\

code 3 & $\{63, 94, 107, 151, 180, 202, 209\}$ & 2608 & 127740 & 65536\\

code 4 & $\{63, 92, 107, 155, 167, 201\}$  & 4144 & 264960 & 589824\\
\hline
\end{tabular}}
\caption{\added{The number of codewords in [256,94] codes with weight less than 32}}
\label{tab2}
\end{center}
\end{table}

\end{example}

\begin{example}
\added{Figure \ref{fig1} illustrates the union bounds \cite{Sason2006} calculated by the weight distribution less than $2w_{\min}$ and the performance of SCL decoding with list size 8 for four different [128, 80] polar codes as outlined in Table \ref{tab1}. In these cases, the performance of SCL8 decoding is close to the ML decoding. Here the union bounds are calculated as $\sum_{d<2w_{\min}}A_d Q(-\sqrt{d}/\sigma)$ where $\sigma$ is the standard variance of additive white Gaussian noise and $Q(x) = \frac{1}{\sqrt{2\pi}}\int_x^{\infty} e^{-z^2/2} dz$ is the distribution function of the standard normal distribution. The simulation results show that the union bound calculated by our results is closed to the performance of ML decoding, particularly at high signal-to-noise ratios (SNR).}

\begin{figure}[!t]
\centering
\includegraphics[width=0.5\textwidth, trim = 80 220 80 220, clip]{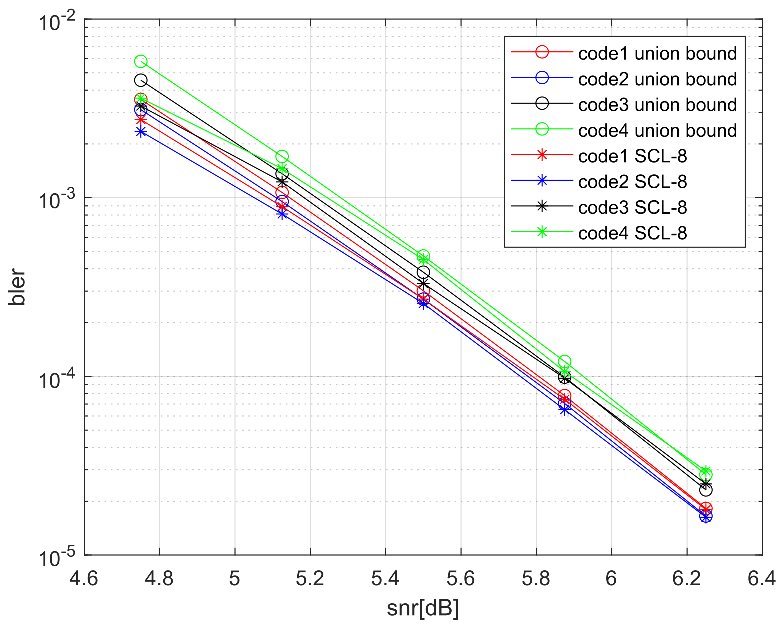}
\caption{\added{SCL performance and union bound for different [128, 80] polar codes}}
\label{fig1}
\end{figure}
\end{example}

\section{Conclusion}

In this paper, we present a comprehensive framework for classifying and enumerating the codewords with weights less than $2w_{\min}$ in decreasing polar codes. We provide the closed-form expressions and the algorithms to calculate the number of these codewords. Importantly, the time complexity of the algorithms is in polynomial with respect to the code length. Our work contributes to a better understanding of the algebraic structure for polar codes and may provide valuable insights that can potentially guide practical constructions in the future.

\appendices
\section*{Appendix A \\ Proof of Lemma \ref{lemma2}}

\begin{proof}
Apparently, each \added{non-constant} polynomial $f$ \added{with linear factors} can be expressed as $g_1 \cdots g_k g$ where $g_i$ are linear polynomials and $g$ has no linear factors. First we prove that $f$ can be rewritten to satisfy the conditions a) and b).

Assume $g_k$ is a linear combination of $g_1,\cdots,g_{k-1}$, i.e., $g_k = \sum_{i=1}^{k-1}a_ig_i + a_0$. Since $g_1(\sum_{i=1}^{k-1}a_ig_i + a_0) = g_1(a_1 + \sum_{i=2}^{k-1}a_ig_i + a_0)$, $g_1 \cdots g_k =  g_1 \cdots g_{k-1}\sum_{i=0}^{k-1}a_i$. As $f$ is non-zero, $\sum_{i=0}^{k-1}a_i = 1$ and $g_1\cdots g_k = g_1\cdots g_{k-1}$. Therefore, without loss of generality we can assume that $g_1,\cdots,g_k$ are linearly independent.

Assume $F(g_i) = F(g_j) = t$ for some $1\leq i<j\leq k$. According to Lemma \ref{lemma1}, we replace $x_t$ in $g_j$ by $g_i + x_t + 1$. The new polynomial is $g_j' = g_i+g_j+1$, then $g_ig_j = g_ig_j'$ and $F(g_j')<F(g_j)$. Note that this procedure does not break the condition a). As $\sum_{i=1}^k F(g_i)$ cannot decrease infinitely, after finite procedures, $F(g_1),\cdots,F(g_k)$ will be all different from each other. Then we reorder $g_1,\cdots,g_k$ so that $F(g_1)<\cdots<F(g_k)$.

Assume $x_t = F(g_i)$ appears in some $g_j$ with $j> i$, by Lemma \ref{lemma1}, replacing $x_t$ in $g_j$ by $g_i + x_t + 1$ does not change the value of $g_ig_j$. Note that this procedure does not break the condition a) or change $F(g_j)$. Similar procedure can also be applied to $g$ to avoid the appearance of $x_{F(g_i)}$.  In this way, $x_t$ only appears in $g_i$, and thus condition b) satisfies.

Next, we prove the second part of the lemma by induction. When $k=0$, since $g$ has no linear factor, we have $j=k$ and $g=h$. 

For the inductive step $k-1\to k$, we replace $x_{F(g_i)}$ by $g_i + x_{F(g_i)} + 1$ for all $i \leq k$ on both sides. After this replacement, $g_i$ becomes 1, $g$ does not change since $x_{F(g_i)}$ does not appear in $g$, and $h_t$ and $h$ becomes $h_t'$ and $h'$ respectively. Then by Lemma \ref{lemma1+}, $g =  h_1' \cdots h_j' \cdot h'$. Since $g$ has no linear factors, $h_1' \cdots h_j' = 1$. 

Suppose $b\in\FF_2^m$ satisfying $g_1(b)\cdots g_k(b) = 1$, then $b_{F(g_i)}=g_i(b) + b_{F(g_i)} + 1$, it follows that $h_t(b)= h'_t(b)$ and then $h_1(b) \cdots h_j(b) = 1$. Similarly, $h_1(b') \cdots h_j(b') = 1$ implies $g_1(b') \cdots g_k(b') = 1$ for $b'\in\FF_2^m$. Hence, $h_1 \cdots h_j = g_1 \cdots g_k$. 

Since deg$(h_1 \cdots h_j) = $ deg$(g_1 \cdots g_k)$, we deduce that $j=k$. Assume $F(h_k) > F(g_k)$, then $x_{F(g_k)}$ appears in $h_1 \cdots h_k$ but not in $g_1 \cdots g_k$, which means the equality cannot hold. Similarly, $F(h_k) < F(g_k)$ is also impossible. Therefore, $F(h_k) = F(g_k)$. 

Since $x_{F(g_k)}$ only appears in $g_k$ and $h_k$, $f =x_{F(g_k)}g_1\cdots g_{k-1} g + g'' = x_{F(g_k)}h_1 \cdots h_{k-1} h + h''$, where $x_{F(g_k)}$ does not appear in $g''$ or $h''$. Therefore, we have $g_1\cdots g_{k-1} g =h_1 \cdots h_{k-1} h$. By the inductive hypothesis, $g_i = h_i$ for $1\leq i\leq k-1$ and $g=h$. Similarly, expanding $f$ in terms of $x_{F(g_{k-1})}$ implies $g_k = h_k$.

\end{proof}

\section*{Appendix B  \\ Proof of Lemma \ref{lemma3}} 

\begin{proof}
Denote $a = \sum_{i=1}^{2\mu} F(f_i)$ to be the sum of indices of the largest terms.

First, if $F(f_{2 i-1} ) = F(f_{2 i}) = t$ for some $1\leq i\leq \mu$, as in Lemma \ref{lemma1}, we replace $f_{2 i}$ by $f_{2 i-1} + f_{2 i} + 1$. This procedure does not change the value of $f$ or break the linear independence condition, and $a$ will decrease. 

Second, if $F(f_{2i-1}) = F(f_{2j-1}) = t$ for some $i\neq j$, denote $f_{2i-1} = x_t + g$ and $f_{2j-1} = x_t + h$. Without loss of generality, let $F(f_{2j})\leq F(f_{2i})$. Since
\begin{equation}
(x_t+g)f_{2i} + (x_t+h)f_{2j} = (x_t+g)(f_{2i}+f_{2j})+(g+h)f_{2j},
\end{equation}
we can rewrite $f$ so that $a$ decreases. Similarly, this procedure does not change the value of $f$ or break the linear independence condition. 

When $F(f_{2i-1}) = F(f_{2j}), F(f_{2i}) = F(f_{2j-1})$, or $F(f_{2i}) = F(f_{2j})$, the procedure is similar.

As $a$ cannot decrease infinitely, after finite procedures described above, $f$ can be rewritten to satisfy that all the largest terms are different.

\end{proof}

\section*{Appendix C  \\ Proof of Lemma \ref{lemma4}} 

\begin{proof}
We prove this lemma by induction. When $\mu = 1$, the lemma is implied by Lemma \ref{lemma2}. For the induction step $\mu-1\to\mu$, if $F(f_{2\mu}) > F(g_{2\mu})$, then $x_{F(f_{2\mu})}$ appears in the left side but not in the right side, which means the equality cannot hold. Similarly, $F(f_{2\mu}) < F(g_{2\mu})$ is also impossible. Therefore, $F(f_{2\mu}) = F(g_{2\mu}) $. Since $x_{F(f_{2\mu})}$ does not appear in other linear polynomials, we have $f_{2\mu-1} = g_{2\mu-1}$. Denote $F(f_{2\mu-1}) = t$.

Define $f'_j$ to be the polynomial by replacing every $x_t$ in $f_j$ by $f_{2\mu-1} + x_t$, and $g'_j$ is defined similarly. Then $g'_{2\mu-1} = f'_{2\mu-1} = 0$. By Lemma \ref{lemma1+}, $\sum_{i=1}^{\mu-1} f'_{2 i-1} f'_{2 i} = \sum_{i=1}^{\mu-1} g'_{2 i-1} g'_{2 i}$. Since $x_t$ cannot be the largest term of other linear polynomials, replacing $x_t$ by $f_{2\mu-1} + x_t$ does not change the largest terms of other polynomials. Then by induction hypothesis, $F(f'_j) = F(g'_j)$ for $1\leq j\leq 2\mu-2$. Therefore, $F(f_i) = F(g_i)$ for all $1\leq i\leq 2\mu$.
\end{proof}

\section*{Appendix D  \\ Proof of Theorem \ref{thm:m1}} 

\added{We first prove the following lemma.}

\begin{lemma}\label{lemma7}
\added{Suppose $\sum_{t=1}^{\mu} h_tg_t $ is proper and satisfies the conditions in Lemma \ref{lemma4}, then we can rewrite $g_t = \sum_{1\leq j\leq m}b^{(t)}_j x_j +b^{(t)}_0$ and $h_t = \sum_{1\leq j\leq m}c^{(t)}_j x_j + c^{(t)}_0$, such that the conditions in Lemma \ref{lemma4} still hold, in addition, $\forall 1 \leq t,t' \leq \mu$,}

\added{1) $b^{(t)}_{F(g_{t'})} = 0$, if $F(g_t)>F(g_{t'})> F(h_{t'}) >F(h_t)$;}

\added{2) $b^{(t)}_{F(h_{t'})} = 0$, if $F(g_t)>F(h_{t'})>F(h_t)$;}

\added{3) $b^{(t')}_{F(h_{t'})} = 0$ ($\sum_{t=1}^{\mu} h_tg_t $ is still proper).}

\end{lemma}

\begin{proof}
\added{First, if $b^{(t)}_{F(g_{t'})} = 1$ for some $F(g_t)>F(g_{t'})> F(h_{t'}) >F(h_t)$, since}
\begin{equation}
\added{h_tg_t + h_{t'}g_{t'} = h_t(g_t+g_{t'}) + (h_t+h_{t'})g_{t'},}
\end{equation}
\added{we can replace $g_t$ by $g_t+g_{t'}$ and $h_{t'}$ by $h_t+h_{t'}$ such that $b^{(t)}_{F(g_{t'})} = 0$ without changing $b^{(t)}_s$ for $s>F(g_{t'})$ and $c^{({t'})}_s$ for $s>F(h_t)$, so the largest terms of $g_t$ and $h_t$ are not changed, thus the conditions in Lemma \ref{lemma4} still hold.}

\added{Similarly, if $b^{(t)}_{F(h_{t'})} = 1$ for some $F(g_t)>F(h_{t'})>F(h_t)$, since}
\begin{equation}
\added{h_tg_t + h_{t'}g_{t'} = h_t(g_t+h_{t'}) + h_{t'}(g_{t'}+h_t),}
\end{equation}
\added{we can replace $g_t$ by $g_t+h_{t'}$ and $g_{t'}$ by $h_t+g_{t'}$ such that $b^{(t)}_{F(h_{t'})} = 0$ without changing $b^{(t)}_s$ for $s>F(h_{t'})$ and $b^{({t'})}_s$ for $s>F(h_t)$.}

\added{If $b^{(t')}_{F(h_{t'})} = 1$, since}
\begin{equation}
\added{h_{t'}g_{t'} = h_{t'}(g_{t'}+h_{t'} + 1),}
\end{equation}
\added{we can replace $g_{t'}$ by $g_{t'}+h_{t'}+1$ such that $b^{(t')}_{F(h_{t'})} = 0$ without changing $b^{(t')}_s$ for $s>F(h_{t'})$.}

\added{Since all the above procedures do not influence $b^{(t)}_s$ with larger subscript, we sort $F(g_{t'})$ and $F(h_{t'})$ for all $1\leq {t'}\leq \mu$ in descending order. Then for all $1\leq t\leq \mu$, check whether conditions 1)-3) satisfy, if not, apply the above procedures (12)-(14). Note that the previously satisfied conditions remain intact throughout the subsequent steps, we can rewrite $\sum_{t=1}^{\mu} h_tg_t$ to satisfy conditions 1)-3).}

\end{proof}

\added{Now we are going to calculate $|A_u|$. By Lemma \ref{lemma2} and Lemma \ref{lemma7}, in order to avoid counting the same polynomials multiple times, we can fix some coefficients to zeros, thus the number of polynomials in $A_u$ is at most 2 to the power of the number of the remaining coefficients, which is $2^{\alpha_u+\beta_u+2\gamma_u}$ as shown in Theorem \ref{thm:m1}. Now we prove that any value of these coefficients yields a different polynomial. Therefore, the number of low-weight codewords is exactly $2^{\alpha_u+\beta_u+2\gamma_u}$.}

\begin{proof}[Proof of Theorem \ref{thm:m1}]
Denote $u = (i_1,\cdots,i_{r+2 \mu-2})\in S_{\mu}$. We calculate the number of different polynomials $f = P(f_1,\dots, f_{r-2}, h_1,g_1, \dots, h_{\mu}, g_{\mu})\in A_u$. Since $f$ is proper, there are $ 2^{\lambda_{x_{i_1}\cdots x_{i_{r-2}}}}$ choices for $f_1,\dots,f_{r-2}$. 

Next we calculate the number of choices for $g = \sum_{t=1}^{\mu} h_t g_t$. Let us write $g = \sum_{1\leq j< k\leq m}a_{j,k}x_jx_k + \sum_{1\leq k\leq m}a_{0,k}x_k + a_{0,0}$, $g_t = \sum_{1\leq s\leq m}b^{(t)}_s x_s + \replaced{ b^{(t)}_0}{b_0}$ and $h_t = \sum_{1\leq s\leq m}c^{(t)}_s x_s + \replaced{ c^{(t)}_0}{c_0}$.  \added{Since $g = \sum_{t=1}^{\mu} h_t g_t$ and $x_j^2 =  x_j$, we have}
\begin{equation}\label{eq12}
\added{a_{j,k} = \sum_{t=1}^{\mu}(b^{(t)}_jc^{(t)}_k+b^{(t)}_kc^{(t)}_j)}
\end{equation}
\added{for $1\leq j< k\leq m$ and} 
\begin{equation}\label{eq13}
\added{a_{0,k} = \sum_{t=1}^{\mu}(b^{(t)}_kc^{(t)}_k+b^{(t)}_kc^{(t)}_0+b^{(t)}_0c^{(t)}_k)}
\end{equation}
\added{for $0\leq k\leq m$.} 

\added{Some of the coefficients in $g_t$ and $h_t$ are prefixed. Firstly, $b^{(t)}_{F(g_t)} = 1$ and $b^{(t)}_s  = 0$ for $s>F(g_t)$, $c^{(t)}_{F(h_t)} = 1$ and $c^{(t)}_s  = 0$ for $s>F(h_t)$.  Besides, since $f$ is proper, $b^{(t)}_s = 0$ if $s=F(h_t)$ or $s=F(f_l)$ for $1\leq l\leq r-2$, $c^{(t)}_s = 0$ if $s=F(f_l)$ for $1\leq l\leq r-2$. Furthermore, we assume the coefficients in $g_t$ which satisfy the conditions in Lemma \ref{lemma7} are fixed to zeros.}

\added{We are going to assign values to the remaining $b^{(t)}_s$ and $c^{(t)}_s$ in a specific order such that for each remaining $b^{(t)}_s$ and $c^{(t)}_s$, there exists one different coefficient $a_{j,k}$ determined by it, thus different values of $b^{(t)}_s$ and $c^{(t)}_s$ yield different polynomials.}\deleted{We will assign values to $b^{(t)}_s$ and $c^{(t)}_s$ in a specific order.} First, we assign values to $b^{(t)}_s$ and $c^{(t)}_s$ \added{appearing in the expanded formula of $a_{j,m}$ from $j=m-1$ to 0} \deleted{in order to determine $a_{j,m}$ for $1\leq j\leq m$ that are relevant to $x_m$}. Then, we assign values \added{to other $b^{(t)}_s$ and $c^{(t)}_s$} appearing in the expanded formula of $a_{j,m-1}$ \replaced{from $j=m-2$ to 0}{for $1\leq j\leq m-1$ relevant to $x_{m-1}$}, and so on. In this way, some $a_{j,k}$ can be freely determined by assigning specific values to $b^{(t)}_s$ and $c^{(t)}_s$ as desired, while other $a_{j,k}$ are determined by other coefficients in $f$ determined earlier. The number of choices for $g = \sum_{i=1}^{\mu} h_i g_i$ is exactly the 2 to the power of the number of free coefficients.

\added{Next, we specify the stage when we assign $b_s^{(t)}$ and $c_s^{(t)}$ in the predefined order. 
If $s > F(h_t)$, $b^{(t)}_s$ first appears in the coefficient of $x_{F(h_t)}x_s$. Thus, $b^{(t)}_s$ is assigned when determining $a_{F(h_t), s}$. If $s < F(h_t)$, $b^{(t)}_s$ is assigned when determining $a_{s, F(h_t)}$. And $c^{(t)}_s$ is assigned when determining $a_{s, F(g_t)}$, since $c^{(t)}_s$ first appears in the coefficient of $x_sx_{F(g_t)}$.}

\added{Now we prove that the first appearance of the free coefficients $b^{(t)}_s$ and $c^{(t)}_s$ are at different stages, that is, they are assigned to determine different $a_{j,k}$. It is clear that all the $b^{(t)}_s$ and $c^{(t)}_{s'}$ are assigned at different stages. Assume that for $t\neq t'$, $b^{(t)}_s$ and $b^{(t')}_{s'}$ are assigned when determining the same $a_{j,k}$. If $s>F(h_t)$ and $s'>F(h_{t'})$, the subscripts of $a_{F(h_t), s}$ and $a_{F(h_{t'}), s'}$ are the same, thus $F(h_t) = F(h_{t'})$, which is against that the largest terms are different. Similarly, the case $s<F(h_t)$ and $s'<F(h_{t'})$ is impossible. W.l.o.g, assume $s>F(h_t)$ and $s'<F(h_{t'})$, we have $j = F(h_t) = s'$, $k = s = F(h_{t'})$. Therefore,
$$
F(h_t) = s' < F(h_{t'}) = s < F(g_{t}).
$$
Since we fix the coefficients in $g_t$ which satisfy the conditions 2) in Lemma \ref{lemma7} to zeros, $b^{(t)}_{s} = b^{(t)}_{F(h_{t'})} = 0$ needs not be considered in our procedure. Therefore, all $b^{(t)}_s$ and $b^{(t')}_{s'}$ are assigned at different stages.}

\added{Assume that for $t\neq t'$, $b^{(t)}_s$ and $c^{(t')}_{s'}$ are assigned when determining the same $a_{j,k}$. If $s>F(h_t)$, the subscripts $j = F(h_t) = s', k = F(g_{t'}) = s$. Therefore,
$$
F(h_t) = s' < F(h_{t'}) < F(g_{t'}) = s < F(g_t). 
$$
Since we fix the coefficients in $g_t$ which satisfy the conditions 1) in Lemma \ref{lemma7} to zeros, $b^{(t)}_{s} = b^{(t)}_{F(g_{t'})} = 0$ needs not be considered in our procedure. If $s<F(h_t)$, then $F(h_t) = F(g_{t'})$, which is against that the largest terms are different. Therefore, $b^{(t)}_s$ and $c^{(t')}_{s'}$ are assigned at different stages.}

\added{Now assume that for $t\neq t'$, $c^{(t)}_s$ and $c^{(t')}_{s'}$ are assigned when determining the same $a_{j,k}$. Therefore, $F(h_t)= F(h_{t'})$, which is against hat the largest terms are different. }

\added{Note that all the free $b^{(t)}_s$ and $c^{(t)}_{s}$ are assigned at different stages, their values determine the corresponding $a_{j,k}$, since other terms in the expanded formula are already assigned. The number of $\sum_{i=1}^{\mu} h_i g_i$ is exactly 2 to the power of the number of free coefficients in $b^{(t)}_s$ and $c^{(t)}_{s}$. Now we count the number of the coefficients: }

\added{1) The number of $c^{(t)}_{s}$ with $s\neq i_1, \cdots, i_{r-2}$ is $\bar{\varphi}_{u,[1, r-2]}(i_{r-3+2t})$. }

\added{2) The number of $b^{(t)}_{s}$ with $s<F(h_t)$ and $s\neq i_1, \cdots, i_{r-2}$ is $\bar{\varphi}_{u,[1, r-2]}(i_{r-3+2t})$.}

\added{3) The number of $b^{(t)}_{s}$ with $s>F(h_t)$ and $s\neq i_l, 1\leq l\leq r+2\mu-2$ is $\bar{\varphi}_{u, [1,r+2\mu-2]}(i_{r-2+2t}) - \bar{\varphi}_{u, [1,r+2\mu-2]}(i_{r-3+2t}))$.}

\added{4) The number of $b^{(t)}_{s}$ with $s>F(h_t)$ and $s = F(g_k)$ is the number of $k$ satisfying $F(g_t)>F(g_k)> F(h_k)> F(h_t)$.} 

In conclusion, the number of \added{free} coefficients \deleted{which can be determined arbitrarily} is \added{$\ |\{(k,t)\mid 1\leq k<t\leq \mu, i_{r-2+2k}>i_{r-2+2t}> i_{r-3+2k}>i_{r-3+2t} \}|+ \sum_{t=1}^{\mu}(\bar{\varphi}_{u, [1,r+2\mu-2]}(i_{r-2+2t}) - \bar{\varphi}_{u, [1,r+2\mu-2]}(i_{r-3+2t})) + 2\sum_{t=1}^{\mu} \bar{\varphi}_{u,[1, r-2]}(i_{r-3+2t}) = $} $\alpha_u + \beta_u +2\gamma_u$, which implies Equation (\ref{eq4}). Equation (\ref{eq5}) follows from Equation (\ref{eq4}) and Theorem \ref{thm2}.
\end{proof}

\section*{Appendix E  \\ Proof of Theorem \ref{thm:m2}}

\begin{proof}
Denote $u = (i_1,\cdots,i_{r+\mu})\in T_{\mu}$. From the proof of Theorem \ref{thm4}, $Q(f_1 \cdots  f_{r+\mu}) = Q(g_1 \cdots g_{r+\mu})$ if and only if $f_i=g_i$ for all $1\leq i\leq r+\mu$. There are $\lambda_{x_{i_1}\cdots x_{i_r}}$ choices for $f_1,\cdots, f_r$. Next we discuss the choice of $f_{r+1},\cdots, f_{r+\mu}$ with largest terms $i_{r+1},\cdots,i_{r+\mu}$ such that $f = Q(f_1 \cdots  f_{r+\mu})\in C(\MI)$ is proper and the linear independence condition is satisfied. We first assign the coefficients in $f_{r+1}$, then $f_{r+2}$ and so on. To guarantee the linear independence condition, we need to ensure that
\begin{equation}\label{eq:3.4.1}
f_{r+k} = \sum_{t=1}^{r-\mu} a_t f_t + \sum_{t =r-\mu+1}^{r} a_t f_t + \sum_{t =r+1}^{r+k-1} a_t f_t + a_0.
\end{equation}
does not hold for any $a_i\in \FF_2$.

Now we consider two cases for $u$:

1) $i_{r-\mu+j} \neq i_{r+j}$ for some $1\leq j\leq \mu$. In this case,  $Q(f_1 \cdots  f_{r+\mu})\in C(\MI)$ if and only if $x_{i_1}\cdots x_{i_r}$ and $x_{i_1}\cdots x_{i_{r-\mu}} x_{i_{r+1}}\cdots x_{i_{r+\mu}}$ belong to $\MI$. Next, we calculate the choices of $f_{r+1},\cdots, f_{r+\mu}$ satisfying the linear independence condition. 

Assume we have already determined $f_{r+j}$ for $1\leq j < k$. There are two cases for $f_{r+k}$:
 
I) $i_{r+k} \neq i_{r-j}$ for all $0\leq j\leq \mu-1$. Assume, by contradiction, Equation (\ref{eq:3.4.1}) holds. 

For $1\leq t\leq r-\mu$, since $x_{i_t}$ only appears in $f_t$, $a_t$ must be 0\added{, so we can delete them from the equation:}

$$
\added{f_{r+k} = \sum_{t =r-\mu+1}^{r} a_t f_t + \sum_{t =r+1}^{r+k-1} a_t f_t + a_0.}
$$

For $r-\mu+1 \leq t\leq r$ and $ i_t > i_{r+k}$, since $x_{i_t}$ does not appear in other $f_l$ with $r-\mu+1\leq l\leq r+k-1$, $a_t$ must be 0. \added{Since $i_{r+k} \neq i_{r-j}$ for all $0\leq j\leq \mu-1$, the equation becomes:}

$$
\added{f_{r+k} = \sum_{t =r-\mu+1, \atop i_t < i_{r+k}}^{r} a_t f_t + \sum_{t =r+1}^{r+k-1} a_t f_t + a_0.}
$$

\replaced{Now all the $f_t$ remained in the equation satisfy}{For other $1 \leq t\leq r+k-1$, we have} $i_t < i_{r+k}$, and $f_{r+k}$ cannot be linear combination of these linear polynomials. 

Therefore, we conclude that $f_{r+k}$ cannot be a linear combination of any $f_j$ where $1\leq j\leq r+k-1$, so the linear independence condition is always satisfied in this case.

The coefficients of $x_{i_j}$ with $1\leq j\leq r-\mu$ or $ r+1\leq j< r+k$ in $f_{r+k}$ must be 0 to satisfy the proper condition. \replaced{Since $\varphi_{u,[r+1,r+k-1]}(i_{r+k}) = k-1$}{Thus}, there are $2^{\added{\bar{\varphi}_{u,[1,r-\mu]}(i_{r+k})}- k + 1}$ choices for $f_{r+k}$.

II) $i_{r+k} = i_{r-j}$ for some $0\leq j\leq \mu-1$. Assume Equation (\ref{eq:3.4.1}) holds. 

Similarly,  $a_t$ must be 0 when $1\leq t\leq r-\mu$, or $r-\mu+1 \leq t\leq r$ with $ i_t > i_{r+k}$. \added{Since $i_{r+k} = i_{r-j}$, the equation becomes:}

$$
\added{f_{r+k} = a_{r-j}f_{r-j} + \sum_{t =r-\mu+1, \atop i_t < i_{r+k}}^{r} a_t f_t + \sum_{t =r+1}^{r+k-1} a_t f_t + a_0.}
$$

Now $x_{i_{r+k}}$ only appears in $f_{r+k}$ and $f_{r-j}$, so $a_{r-j}$ must be 1. For $r-\mu+1 \leq t\leq r$ and $ i_t < i_{r+k}$, $a_t$ can be assigned arbitrarily. And for $r-1 \leq t\leq r+k-1$, since the form is proper, $a_t$ must be assigned so that $x_{i_t}$ does not appear in $f_{r+k}$. 

Therefore, the free coefficients are $a_0$ and $a_t$ with $r-\mu+1 \leq t\leq r$ and $ i_t < i_{r+k}$. Since all the polynomials are linear independent, there are $2^{\added{\varphi_{u,[r-\mu+1,r]}(i_{r+k})}+1}$ choices for $f_{r+k}$ such that Equation (\ref{eq:3.4.1}) holds.

Therefore, there are $s_u(k) = 2^{\added{\bar{\varphi}_{u,[1,r-\mu]}(i_{r+k})}- k + 1}-2^{\varphi_{u,[r-\mu+1,r]}(i_{r+k})+1}$ choices for $f_{r+k}$. If $\added{\bar{\varphi}_{u,[1,r-\mu]}(i_{r+k})}- k \leq \added{\varphi_{u,[r-\mu+1,r]}(i_{r+k})}$, then \added{there does not exist any $f_{r+k}$ such that linear independence holds, so} $B_u = \varnothing$.

2) $i_{r-\mu+j} = i_{r+j}$ for all $1\leq j\leq \mu$. It should be noted that \added{in this case} $x_{i_1}\cdots x_{i_r}$ may not belong to $\MI$. Then $f_{r+1},\cdots,f_{r+\mu}$ must be chosen to be linearly independent and $Q(f_1,\cdots, f_{r+\mu})\in C(\MI)$. We first analyse $f_{r+1}$, then $f_{r+2}$ and so on. 

To guarantee $f\in C(\MI)$, notice that the coefficients larger than $b_u(k)$ in $f_{r+k}$ and $f_{r-\mu+k}$ must be equal, and then the monomials that do not belong to $\MI$ in $f_{r-\mu}\cdots f_r$ and $f_{r+1}\cdots f_{r+\mu}$ will cancel each other out. Therefore, the number of free coefficients of $f_{r+k}$ is $\added{\bar{\varphi}_{u,[1,r]}(b_u(k))}+1$. Next we analyse the number of linearly dependent choices. Similar to the first case, assume Equation (\ref{eq:3.4.1}) holds. Since $x_{i_{r+j}}$ with $1\leq j < k$ only appears in $f_{r+j}$ and $f_{r+j-\mu}$, we have $a_{r+j} = a_{r-\mu+j}$. Then there are $2^k$ choices for $f_{r+k}$ such that Equation (\ref{eq:3.4.1}) holds.

Therefore, there are $s_u(k) = 2^{\added{\bar{\varphi}_{u,[1,r]}(b_u(k))}+1}-2^k$ choices for $f_{r+k}$. If $\added{\bar{\varphi}_{u,[1,r]}(b_u(k))} + 1\leq k$, $B_u = \varnothing$. 

In conclusion, since there are $\lambda_{x_{i_1}\cdots x_{i_r}}$ choices for $f_1,\cdots, f_r$ and $\sigma_u$ choices for $f_{r+1},\cdots,f_{r+\mu}$, we prove Equation (\ref{eq6}). Equation (\ref{eq7}) follows from Equation (\ref{eq6}) and Theorem \ref{thm4}.
\end{proof}

\end{document}